\documentclass[preprint,11pt]{elsarticle}
\usepackage{latexsym,amsthm,amsmath,amssymb}
\usepackage[utf8x]{inputenc}

\newtheorem{theorem}{Theorem}
\newtheorem{lemma}{Lemma}

\usepackage{mathrsfs}
\usepackage{amsmath}
\usepackage{amsthm}
\usepackage{amssymb}

\usepackage{ifthen}

\usepackage{intcalc}
\usepackage{tikz}
\usetikzlibrary{arrows}

\newtheorem{proposition}{Proposition}
\newtheorem{corollary}{Corollary}
\newtheorem{definition}{Definition}

\newtheorem{claim}{Claim}

\newtheorem{conjecture}{Conjecture}

\newcommand{\edn}{\gamma^{\infty}}
\newcommand{\ednm}{\gamma^{\infty}_m}
\newcommand{\oedn}{\overrightarrow{\edn}}
\newcommand{\oednm}{\overrightarrow{\ednm}}
\newcommand{\oalpha}{\overrightarrow{\alpha}}
\newcommand{\dist}{\textrm{dist}}
\newcommand{\diam}{\textrm{diam}}
\newcommand{\scdd}{\gamma_\leftrightarrow}
\newcommand{\kscdd}{\scdd'}
\newcommand{\oscdd}{\overrightarrow{\scdd}}

\newcommand{\isorient}{\textrm{ is an orientation of  }}
\newcommand{\ffceil}[1]{\left \lceil #1 \right \rceil}
\newcommand{\fffloor}[1]{\left \lfloor #1 \right \rfloor}

\begin{document}

\begin{frontmatter}

\title{Eternal dominating sets on digraphs and orientations of graphs}

\author[GOAL]{Guillaume Bagan}
\author[GOAL]{Alice Joffard}
\author[GOAL]{Hamamache Kheddouci}

\address[GOAL]{Lab. LIRIS, UMR CNRS 5205, University of Lyon, F-69003\\University of Claude Bernard Lyon 1\\ 43 Bd du 11 Novembre 1918, F-69622, Villeurbanne, France.}

%\cortext[cor1]{Corresponding Author, guillaume.bagan@liris.cnrs.fr}

\begin{abstract}
We study the eternal dominating number and the m-eternal dominating number on digraphs. We  generalize known results on graphs to digraphs.
We also consider the problem "oriented (m-)eternal domination", consisting in finding an orientation of a graph that minimizes its eternal dominating number.
We prove that computing the oriented eternal dominating number is NP-hard and characterize the graphs for which the oriented m-eternal dominating number is 2.
We also study these two parameters on trees, cycles, complete graphs, complete bipartite graphs, trivially perfect graphs and different kinds of grids and products of graphs.
\end{abstract}

\end{frontmatter}

\section{Introduction}

\subsection{Definitions and notations}

In this paper, graphs and digraphs are considered finite, without multiple edges (arcs) and without loops. However, a digraph can contain both arcs $(u, v)$ and $(v, u)$ for some vertices $u$ and $v$.
When there is no ambiguity on the (di)graph $G$, $V$ is the vertex set of $G$, $E$ is the edge (arc) set of $G$,  $n$ is the order of $G$ and $m$ is the number of edges in $G$.

Let $G = (V, E)$ be a graph.
Given a vertex $v \in V$, $N(v)$ represents the open neighborhood of $v$ i.e. the set $\{ y \in V: xy \in E \}$
and the closed neighborhood of $v$ is $N[v] = N(v) \cup \{ v \}$.
A set $S \subseteq V$ is an independent set of $G$ if there is no edge $uv$ for every $u, v$ in $S$.
The independent set number $\alpha(G)$ is the size of a largest independent set of $G$.
A set $S \subseteq V$ is a dominating set of $G$ if
$\bigcup_{v \in S} N[v] = V$.
The dominating number $\gamma(G)$ is the size of a smallest dominating set of $G$.
A set $S \subseteq V$ is a $k$-dominating set of $G$ if
for every vertex $v \in V \setminus S$, $|N(v) \cap S| \geq k$.
$G$ is $k$-vertex-connected (resp. $k$-edge-connected) if $G$ stays connected if one removes any set of at most $k-1$ vertices (resp. edges).
$\gamma_{k,l}(G)$ is the size of the smallest $k$-dominating $l$-edge-connected set of $G$.

Let $G = (V, E)$ be a digraph.
For two vertices $u$ and $v$ of $G$, 
we say that $u$ dominates $v$ if there is an arc $(u, v)$ in $G$. $N^+(u) = \{ v : (u, v) \in E  \}$ 
and $N^+[u] = N^+(u) \cup \{ u \}$.
A set $S \subseteq V$ is a dominating set of $G$ if
$\bigcup_{v \in S} N^+[v] = V$.
The distance $\dist(u, v)$ from $u$ to $v$ is the number of edges of the shortest path from $u$ to $v$. Given a set $S \subseteq V$, $\dist(u, S) = \min \{ dist(u, v) : v \in S\}$.
If $G$ is strongly connected, the diameter $\diam(G)$ of $G$ is the maximum distance from $u$ to $v$ for every vertex $u$ and $v$ in $V$.

$P_n$ represents a path graph with $n$ vertices, $C_n$ represents a cycle graph with $n$ vertices, $K_n$ represents the complete graph with $n$ vertices and $K_{n,m}$ represents a complete bipartite graph with partitions of size $n$ and $m$.
Given two graphs $G_1 = (V_1, E_1)$ and $G_2 = (V_2, E_2)$,
$G_1 \square G_2$ is the cartesian product of $G_1$ and $G_2$.
$G_1 \boxtimes G_2$ is the strong product of $G_1$ and $G_2$.
$[k]$ represents the set $\{1, 2, \ldots, k\}$.
A wqo $\leq_P$ is a preorder such that any infinite sequence of elements
$x_0, x_1, x_2, \ldots$ contains an increasing pair $x_i \leq_P x_j$ with $i<j$. In particular, it does not admit infinite decreasing sequences.

\subsection{Eternal domination}

The problem of eternal domination on undirected graphs, while being a rather recent problem, has been widely studied (see \cite{klostermeyer2016protecting} for a survey).
It has initially been motivated by problems in military defense.

The eternal domination on a graph $G$ can be seen as an infinite game between two players: the \emph{defender} and the \emph{attacker}.
First, the defender chooses a set $D_0$ of $k$ vertices called the \emph{guards}.
At turn $i$, the attacker chooses a vertex $r_i$ called \emph{attack} in
$V \setminus D_{i-1}$ and the defender must \emph{defend} the attack by moving to $r_i$ a guard on a vertex $v_i$ adjacent to $r_i$.
The new guards configuration is $D_i = D_{i-1} \cup \{r_i \} \setminus \{ v_i \}$.
The defender wins the game if he can defend any infinite sequence of attacks.
The \emph{eternal domination number}, denoted by $\edn(G)$, is the minimum number of guards necessary for the defender to win.
An \emph{eternal dominating set} is a set that can initially be chosen by the defender in a winning strategy.

We now give a more formal definition of these notions.

\begin{definition}
Let $G = (V, E)$ be a graph. The set $EDS(G)$ of eternal dominating sets of $G$ is the greatest set of subsets of $V$ such that
for every $S\in EDS(G)$ and every $r \in V \setminus S$, there is a vertex $v \in S$ such that $\{v,r\} \in E$ and $S\cup \{r\}\setminus \{v\} \in EDS(G)$.

The eternal domination number of $G$ is defined as $\edn(G)=\min\{|S|,S\in EDS(G)\}$.
\end{definition}

A variant of the eternal domination is the m-eternal domination, where the defender is authorized to move several guards at a time. Notice that the "m" of m-eternal does not represent a value.

\begin{definition}
Let $G = (V, E)$ be a graph. 
Given two sets $S_1, S_2 \subseteq V$, a multimove $f$ from $S_1$ to $S_2$ is a one-to-one mapping from $S_1$ to $S_2$ such that
for every $x \in S$, we have $f(x) = x$ or $\{ x, f(x) \} \in E$.

The set $MEDS(G)$ of m-eternal dominating sets of $G$ is the greatest set of subsets of $V$ such that
for every $S\in MEDS(G)$ and every $r \in V(G) \setminus S$, there is a multimove $f$ such that $r \in f(S)$ and $f(S) \in MEDS(G)$.

Let $G$ be a graph. The m-eternal domination number of $G$ is defined as
$\ednm(G)=\min\{|S|,S\in MEDS(G)\}$.
\end{definition}

The two following results compare the value of both the eternal domination number and the m-eternal domination number to the value of other well known graph parameters.

\begin{theorem}\label{eternal-seq-ineqs}\cite{burger2004infinite,goddard2005eternal,kloster2007}
Given a graph $G$, we have
$$\gamma(G) \leq \ednm(G) \leq \alpha(G) \leq \edn(G) \leq \binom{\alpha(G)+1}{2}$$
where $\gamma(G)$ is the domination number of $G$, and $\alpha(G)$ is the  independent set number of $G$.
\end{theorem}

\begin{theorem}\label{eternal_clique_covering}\cite{burger2004infinite}
Given a graph $G$, we have
$$\edn(G) \leq \theta(G)$$ where $\theta(G)$ is the clique covering number of $G$.
\end{theorem}

These two theorems are particularly interesting when $G$ is a perfect graph since $\alpha(G) = \edn(G) = \theta(G)$ by definition of a perfect graph (and its complement).

The values of $\edn$ and $\ednm$ have also been studied for many classes of graphs. A lot of attention has been given to grid graphs.

Another theorem gives an upper bound for $\ednm$.
\begin{theorem}\label{ednm_connecteddominant}
\cite{goddard2005eternal}
Given a graph $G$, we have
$$\ednm(G) \leq \gamma_c(G)+1$$ where $\gamma_c(G)$ is the size of a smallest connected dominating set of $G$.
\end{theorem}

\subsection{Contributions}

To our knowledge, the eternal domination problem has only been studied on undirected graphs. 
In this paper, we consider eternal domination on directed graphs where the guards must follow the direction of the arcs.
 Additionally, as it has been done for many digraph parameters such as diameter, chromatic number, domination number or maximum outgoing degree, we consider the problem, namely oriented (m-)domination, consisting in finding an orientation of an undirected graph which minimizes $\edn$ or $\ednm$.

We present our contributions:
we generalize Theorem \ref{eternal-seq-ineqs} to digraphs.
We also give a generalization of Theorem \ref{ednm_connecteddominant} by introducing the notion of dominating-dominated set.
We give other upper bounds by various parameters.
We show that the oriented eternal domination problem is coNP-hard in general.
We characterize graphs for which the oriented m-eternal number is 2.
We study thr oriented eternal domination and oriented m-eternal domination problems
on some graph classes: trees, complete graphs, complete bipartite graphs, trivially perfect graphs and various kinds of grids. In particular, we introduce the notion of neighborhood-equitable coloring and use it to find upper bounds.

\section{Eternal domination on digraphs}
In this section, we consider the eternal domination and m-eternal domination problems on directed graphs. The definitions of $EDS(G)$, $\edn(G)$, $MEDS(G)$ and $\ednm(H)$ are straightforward, the only difference is that whenever an edge is considered, we consider an arc instead.

From this, we can already deduce several general results for digraphs. 

It is straightforward that the (m-)eternal number of a graph is the sum of the  (m-)eternal number of each of its connected components.
We will prove that this result can be extended to the strongly connected components of a digraph.

\begin{lemma}\label{digraph-subgraph}
Let $G$ be a  digraph with strongly connected components $S_1, \ldots S_l$.
Then $$\edn(G) = \sum_{i=1}^l \edn(G[S_i])$$
and
 $$\ednm(G) = \sum_{i=1}^l \ednm(G[S_i]).$$
\end{lemma}

\begin{proof}
We only prove the first equality. The proof of the second equality is similar.
The statement $\edn(G) \leq \sum_{i=1}^l \edn(G[S_i])$ is straightforward.
Let us prove the other part of the equality.

We define the following relation $\leq_r$ on the vertices of $G$: $x \leq_r y$ if $x$ is reachable from $y$ in $G$.
Since $G$ is finite, $\leq_r$ is a wqo.
Let us represent a guard configuration of $k$ guards by a tuple of $k$ vertices and extend the relation
$\leq_r$ on (ordered) guard configurations: $(x_1, \ldots, x_k) \leq_r (y_1, \ldots, y_k)$ iff $x_i \leq_r y_i$ for every $i \in [k]$. This relation is also a wqo since it is the direct product of $k$ wqos.
It is easily seen that if the guard configuration $c_2$ is reachable from the guard configuration 
$c_1$ then $c_1 \geq_r c_2$.

Assume now that the defender plays with less than $\sum_{i=1}^l \ednm(G[S_i])$ guards.
Let $D_0$ be the initial configuration.
Thus, there is a s.c.c. $S_i$ such that $S_i \cap D_0 < \edn(G[S_i])$.
Consequently, the attacker can apply a strategy in $G[S_i]$ such that he can either win, or the defender moves a guard outside $S_i$ into $S_i$.
In the second case, we obtain a configuration $D_j$ such that $D_0 >_r D_j$.
So, the attacker can apply again its strategy on another s.c.c. $S_{i'}$ with $S_{i'} \cap D_j < \edn(G[S_{i'}])$. Since $\leq_r$ is a wqo, there is no infinite decreasing sequence of configurations. Thus, the attacker eventually wins.
\end{proof}

This leads to the value of the two parameters for directed acyclic graphs.

\begin{corollary}\label{eternal_acyclic}
If $G$ is an acyclic digraph with $n$ vertices, then $\edn(G) = \ednm(G) = n$.
\end{corollary}

An important result concerns the monotony of $\edn$, that is not verified for $\ednm$. This has been proved on graph \cite{klostermeyer2005eternally} but the adaptation on digraphs is straightforward.

\begin{lemma}\label{induced_edn}\cite{klostermeyer2005eternally}
Let $G$ be a digraph and $H$ be an induced subgraph of $G$.
Then, $\edn(H) \leq \edn(G)$.
\end{lemma}

Notice also that $\edn$ and $\ednm$ do not decrease when one removes edges.

The main result of this section is the generalization of Theorem \ref{eternal-seq-ineqs} to
directed graphs.

We define $\alpha$ for digraphs as follows.

\begin{definition}
Given a digraph $G$,
$\alpha(G)$ is the order of the greatest induced acyclic subgraph of $G$.
\end{definition}

Notice that this definition of $\alpha$ is, in some sense, a generalization of the one for undirected graphs. Indeed, if we replace every edge of a graph $G$ by two arcs, thus creating the digraph $\overleftrightarrow{G}$, we have the equality
$\alpha(G) = \alpha(\overleftrightarrow{G})$.

%$\alpha(G)$ for a digraph $G$ denotes the order of the greatest acyclic induced subgraph of $G$.
%This parameter can be seen as a generalization of $\alpha(G)$ for graphs since 
%$\alpha(G) = \alpha()$ for any graph $G$ where $\overleftrightarrow{G}$
%denotes the symmetrical orientation of $G$.

\begin{theorem}\label{eternal-seq-ineqs-digraph}
Given a digraph $G$, we have
$$\gamma(G) \leq \ednm(G) \leq \alpha(G) \leq \edn(G) \leq \binom{\alpha(G)+1}{2}.$$
\end{theorem}

\begin{proof}
The inequality $\gamma(G) \leq \ednm(G)$ is straightforward.
The inequality $\alpha(G) \leq \edn(G)$ is a consequence of Corollary \ref{eternal_acyclic} and Lemma \ref{induced_edn}.
The inequality $\edn(G) \leq \binom{\alpha(G)+1}{2}$ can easily be proved by adapting the proof in \cite{kloster2007}.
The remaining inequality is $\ednm(G) \leq \alpha(G)$.
For this last inequality, the principle is similar to the proof
of Goddard et al. \cite{goddard2005eternal} for undirected graphs, but with notable differences, that make the proof a non trivial adaptation.
We consider two cases, depending on the existence of a vertex $v \in V$ that does not belong to any of the greatest induced acyclic subgraphs of $G$.

If such a vertex does not exist, then every vertex belongs to at least one of the greatest induced acyclic subgraphs of $G$.
We thus prove that for any two distinct subsets $A,B \subseteq V$ of vertices that each induce a greatest induced acyclic subgraph of $G$, there exists a multimove from $A$ to $B$, so that we can always defend an attacked vertex $a$ by going from a subset of vertices inducing a greatest induced acyclic subgraph to another that contains $a$.
Let $G'$ be the bipartite undirected graph $G'=(A,B,E')$ with $E'=\{ xy ,x\in A,y\in B,(x,y)\in E \mbox{ or } x=y \}$. \footnote{Actually $G'$ is not necessarily a bipartite graph and can contain loops because $A$ and $B$ can have common elements. However, Hall theorem is still applicable. Moreover, the graph can become bipartite if we make a copy of each element of $A \cap B$.}
See Figure \ref{ednm_bipartite_example} for an example.

%constructed by taking the vertices of $A$ in one part, the vertices of $B$ in %the other, thus creating a copy of the vertices of $A\cap B$, and by putting an %edge between two vertices that are either linked in $G$  by an arc going from %$A$ to $B$, or between a vertex of $A\cap B$ and its copy.

We show that for any set $S\subseteq A$, $|N(S)| \geq |S|$, where $N(S)$ is the set of neighbors of vertices of $S$ in $G'$. We will thus be able to use Hall's marriage theorem \cite{hall1935representatives} and prove that there exists a perfect matching between $A$ and $B$ in $G'$, and therefore a multimove from $A$ to $B$ in $G$.
Suppose that there exists $S \subseteq A$ such that $|N(S)|<|S|$.
Let $X$, $Y$ and $C$ be the sets of vertices defined by $X=S\setminus B$, $Y=N(S)\setminus A$ and $C=B\cup X \setminus Y$.
We have $|C|>|B|$.
Indeed, by construction, $S \cap B \subseteq N(S\cap B)$, so that $|S \cap B| \leq |N(S)\cap A|$. Thus, since $|S|>|N(S)|$, we have $|X|>|Y|$, which gives $|C|>|B|$.

Moreover, $G[C]$ is acyclic. Indeed, $G[A]$ and $G[B]$ are acyclic, and for any two vertices $x,y \in C$, $\{x,y\} \not \in E'$, so that for all $x\in C\cap A, y\in C\cap B$, $(x,y) \not \in E$.
Thus, $G[C]$ is an induced acyclic subgraph of order greater than $|B|$ and, by  absurd, Hall's condition is respected. Thus, there exists a perfect matching between $A$ and $B$ in $G'$, and a multimove from $A$ to $B$ in $G$.

Now, if there exists a vertex $v \in V$ such that $v$ does not belong to any of the greatest induced acyclic subgraphs of $G$, we use induction to prove that $\ednm(G) \leq \alpha(G)$.
For $n=1$, we obviously have $\ednm(G) \leq \alpha(G)$.
Assume that $n > 1$.
We define $G'$ from $G$ by contracting the vertex $v$ as follows:
$G'=(V',E')$ with $V'=V\setminus \{v\}$ and
$E' = E(G[V']) \cup \{(u,w)\in V'^2, u \neq w \wedge (u,v),(v,w) \in E\}$.
By induction hypothesis, we have $\ednm(G') \leq \alpha(G')$.

We have $\alpha(G')\leq \alpha(G) -1$. Indeed,
consider a set $S \subseteq V'$ of cardinality $\alpha(G)$.
Let us prove that $G'[S]$ is cyclic.
If $G[S]$ is cyclic, $G'[S]$ is cyclic too.
Assume that $G[S]$ is acyclic.
Then, by maximality of $S$, $G[S\cup \{v\}]$ is cyclic.
Either there exists two distinct vertices $u, w \in S$ such that $(u,v),(v,w) \in E$ and $(u,v)$, $(v,w)$ are part of a cycle in $S$, or, otherwise, there exists $u \in S$ such that $ (u,v),(v,u) \in E$. In this last case, since $v$ does not belong to any of the greatest induced acyclic subgraphs of $G$, $G[S \setminus \{u\}]\cup \{v\}]$ is cyclic, so that there exists $w \in S \setminus \{u\}$ such that $(w,v),(v,w) \in E$.
In those two cases, by definition of $G'$, $G'[S]$ is cyclic.
Thus, there exists no induced acyclic subgraph of $G'$ of order $\alpha(G)$, and $\alpha(G')\leq \alpha(G)-1$.

Moreover, $\ednm(G) \leq \ednm(G') + 1$.
Indeed, if we can defend $G'$ with $k$ guards, we can defend $G$ with $k+1$ guards by adding a guard on $v$ and copying the strategy in $G$.
If the defender of $G'$ moves a guard from $u$ to $w$ using an edge $(u, w)$ that has been added in $G'$ (and, thus, not present in $G$), then the two edges $(u, v)$ and $(v, w)$ belong to $G$. Thus the defender of $G$ moves a guard from $u$ to $v$ and another one from $v$ to $w$.

Those three inequalities finally give $\ednm(G) \leq \alpha(G)$.
\end{proof}

\begin{figure}
	\centering
	\begin{tikzpicture}[vertex/.style={circle,draw,fill,inner sep=2pt}]    
    
    %Graphs
    \node[vertex,label=left:$v_1$] (v1) at (0,0) {};
    \node[vertex,label=right:$v_2$] (v2) at (2,2) {};
    \node[vertex,label=right:$v_3$] (v3) at (4,0) {};
    \node[vertex,label=left:$v_4$] (v4) at (0,4) {};
    \node[vertex,label=above:$v_5$] (v5) at (0,6) {};

    \node[vertex,label=left:$v_1$] (a1) at (8,0) {};
    \node[vertex,label=left:$v_2$] (a2) at (8,3) {};
    \node[vertex,label=left:$v_5$] (a5) at (8, 6) {};
    \node[vertex,label=right:$v_2$] (b2) at (12,0) {};
    \node[vertex,label=right:$v_3$] (b3) at (12,3) {};
    \node[vertex,label=right:$v_5$] (b5) at (12,6) {};
    
    \path[->]
     (v1) edge[-triangle 90] (v2)
     (v2) edge[-triangle 90] (v3)
     (v3) edge[-triangle 90] (v1)
     (v1) edge[-triangle 90,bend left=20] (v4)
     (v4) edge[-triangle 90,bend left=20] (v1)
     (v4) edge[-triangle 90] (v2)
     (v3) edge[-triangle 90,bend right=20] (v4)
     (v5) edge[-triangle 90] (v4)
    ;
    
    \draw (a1) -- (b2)
          (a2) -- (b2)
          (a2) -- (b3)
          (a5) -- (b5)
          ;

	\end{tikzpicture}
    \caption{On the left, a graph $G$ for which $G[\{v_1, v_2, v_5\}]$
    and $G[\{v_2, v_3, v_5\}]$ are maximum induced acyclic subgraphs. On the right, the bipartite graph constructed in the proof of Theorem \ref{eternal-seq-ineqs-digraph}.}
	\label{ednm_bipartite_example}
\end{figure}
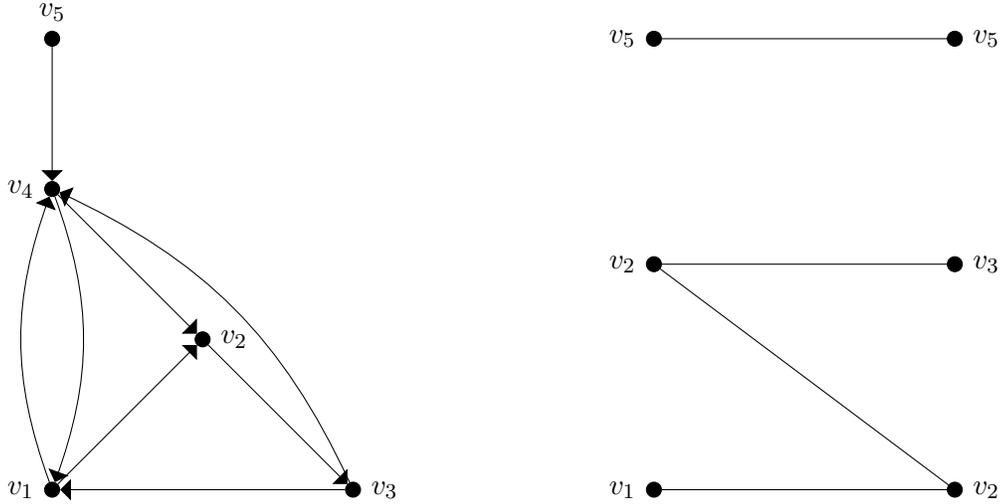

We now generalize Theorem \ref{ednm_connecteddominant} by introducing the notion of dominating-dominated set of a digraph $G$.

\begin{definition}
Let $G = (V, E)$ be a digraph.
A set $S \subseteq V$ is a dominating-dominated set of $G$
if every vertex $v \in V \setminus S$ dominates a vertex of $S$
and is dominated by a vertex of $S$.

$\scdd(G)$ is the size of a smallest strongly connected dominating-dominated set of $G$.
\end{definition}

\begin{theorem}\label{ednm_scdd}
Let $G = (V, E)$ be a digraph.
Then $\ednm(G) \leq \scdd(G) + 1$.
\end{theorem}

\begin{proof}
Let $S$ be a strongly connected dominating-dominated set of $G$.
We will show that we can defend $G$ with $|S|$ + 1 guards.
The invariant for a guard configuration is that it contains all vertices of $S$.
Initially, we put the guards on $S$ and an arbitrary vertex $u$ in
$V \setminus S$.
Consider an attack on a vertex $v$.
By definition of $S$, there exists two (non necessarily distinct) vertices $x, y \in S$ such that $(u, x) \in E$ and $(y, v) \in E$.
Thus, we move the guard on $y$ to $v$.
Let $p$ be a path from $x$ to $y$ using only vertices of $S$.
We move the guard on $u$ to $x$ and each guard on a vertex $w$ of $p$ to the next vertex.
We obtain the guard configuration $S \cup \{ v \}$. Thus the invariant is verified. 
\end{proof}

We generalize this notion to $k$-dominating-dominated sets.

\begin{definition}
Let $G = (V, E)$ be a digraph.
A set $S \subseteq V$ is a $k$-dominating-dominated set of $G$
if every vertex $v \in V \setminus S$ is dominated by a vertex in $S$ and is at distance at most $k$ to a vertex in $S$.

$\kscdd(G)$ is the minimum value of $k + |S|-1$ 
for every $k$ and $S$ with $S$ a strongly connected $k$-dominating-dominated set of $G$.
\end{definition}

$\gamma_c(G)$ denotes the size of the smallest strongly connected
dominating set of $G$.

\begin{theorem}
Let $G$ be a strongly connected digraph.
Then $$\ednm(G) \leq \kscdd(G) + 1 \leq \gamma_c(G) + \diam(G) \leq (\gamma(G) + 1) \diam(G).$$
\end{theorem}

\begin{proof}
To prove the first inequality, we consider 
a $k$-dominating-dominated set $S$ of $G$. We defend $G$ with
$|S| + k$ guards. The invariant is as follows: $|S|$ guards are on $S$ and $k$ guards $(v_1, \ldots, v_k)$ are in $V \setminus S$ with $\dist(v_i, S) \leq i$ for every $i \in k$.
We start with an initial configuration that satisfies the invariant. Assume that a vertex $u$ is attacked.
Since $S$ is a dominating set, there exists a vertex $x \in S$
which dominates $u$. So, we move the guard on $x$ to $u$.
We "push" the guard on $v_1$ to $x$ similarly to the proof of
Theorem \ref{ednm_scdd} and we move each guard on $v_i$, $i \geq 2$, to a vertex $w_i$ with $\dist(w_i, S) \leq i-1$.
The invariant is satisfied.

To prove the second inequality, notice that a dominating set 
is a $\diam(G)$-dominating-dominated set.

To prove the last inequality, consider a dominating set $S = \{v_1, \ldots, v_l\}$. For each $i \in [l-1]$, we add to $S$ the vertices of a shortest path from $v_i$ to $v_{i+1}$.
We do the same between $v_l$ and $v_1$.
We add at most $l(\diam(G) - 1)$ vertices and obtain a set $S'$ of size $l \cdot \diam(G)$ which induces a strongly connected subrgraph.
\end{proof}

\section{Eternal domination on orientations of graphs}

In this section, we are interested in orientating an undirected graph in order to minimize its eternal domination number, or its m-eternal domination number.
An orientation of an undirected graph $G$ is an assignment of exactly one direction to each of the edges of $G$.
This leads to the introduction of three new parameters for undirected graphs:

\begin{definition}Given a (non directed) graph $G$,\\
$\oedn(G) = \min \{ \edn(H) : H \isorient G \}$\\
$\oednm(G) = \min \{ \ednm(H) : H \isorient G \}$\\
$\oalpha(G) = \min \{ \alpha(H) : H \isorient G \}$\\
$\oscdd(G) = \min \{ \scdd(H) : H \isorient G \}$.\\
\end{definition}

Notice that, for non trivial graphs, $\oedn$ can never be equal to $\gamma$.

\begin{proposition}\label{eternal_alpha_olpha}
Let $G$ be a graph with at least one edge.
Then, $\gamma(G) \leq \alpha(G) < \oalpha(G) \leq \oedn(G)$.
\end{proposition}

\begin{proof}
We only need to prove that $\alpha(G) < \oalpha(G)$.
Let $H$ be an orientation of $G$ and $S$ be a maximum independent set of $G$.
Let $S' = S \cup \{v\}$ where $v$ is an arbitrary vertex in $V \setminus S$.
$G[S']$ is an union of stars and isolated vertices. Thus, $H[S']$ is acyclic.
\end{proof}

We conjecture a stronger result.
\begin{conjecture}
Let $G$ be a graph with at least one edge.
Then, $\theta(G) < \oalpha(G)$.
\end{conjecture}
The conjecture would imply that there is no non trivial graph with $\edn = \oedn$.
It is true for perfect graphs since $\alpha(G) = \edn(G) = \theta(G) < \oalpha(G)$ (Theorems \ref{eternal-seq-ineqs} and \ref{eternal_clique_covering}).

Robbins proved the following theorem.

\begin{theorem}\label{robbins_theorem}\cite{robbins1939theorem}
Let $G$ be a graph.
Then $G$ admits a strongly connected orientation if and only if $G$ is 2-edge-connected.
\end{theorem}

This leads to the following proposition.

\begin{proposition}\label{eternal_strongly_connected}
Let $G$ be a 2-edge-connected graph.
Then there exists a strongly connected orientation $H$ of $G$ with 
$\edn(H) = \oedn(G)$ and $\ednm(H) = \oednm(G)$.
\end{proposition}

\begin{proof}
Let $H$ be an orientation of $G$ such that $\ednm(H) = \oednm(G)$. It suffices to notice that in $H$, each guard stays inside its strongly connected component.
Thus, we can change in $H$ the orientation of some edges between two components to make it strongly connected.
\end{proof}

By combining Proposition \ref{eternal_strongly_connected} and Lemma 1, we obtain:

\begin{corollary}\label{eternal_2_vertex_compoenents}
Let $G$ be a  graph with 2-edge-connected components $S_1, \ldots S_l$.
Then $$\oedn(G) = \sum_{i=1}^l \oedn(G[S_i])$$
and
 $$\oednm(G) = \sum_{i=1}^l \oednm(G[S_i]).$$
\end{corollary}

Chambers et al \cite{chambers2006mobile} proved that $\ednm(G) \leq \ffceil{\frac{n}{2}}$
for every connected graph $G$. One natural question is to see if this proposition is true with $\oednm$
for every 2-edge-connected graphs.
We disprove this proposition by giving a counterexample that is even 2-vertex-connected in Figure 2.

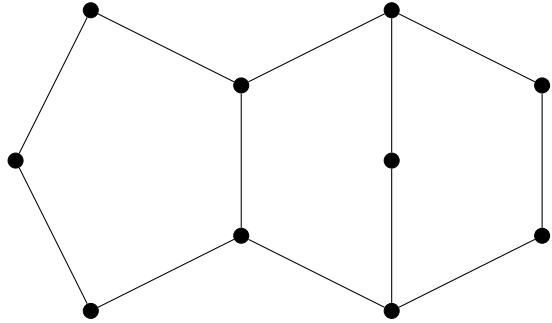
\begin{figure}
	\centering
	\begin{tikzpicture}[vertex/.style={circle,draw,fill,inner sep=2pt}]    
    
    %Graphs
    \node[vertex] (v1) at (0,0) {};
    \node[vertex] (v2) at (-1,2) {};
    \node[vertex] (v3) at (0,4) {};
    \node[vertex] (v4) at (2,3) {};
    \node[vertex] (v5) at (2,1) {};
    
    \node[vertex] (v6) at (4,0) {};
    \node[vertex] (v7) at (4,2) {};
    \node[vertex] (v8) at (4,4) {};
    \node[vertex] (v9) at (6,3) {};
    \node[vertex] (v10) at (6,1) {};
    
    \draw (v1) -- (v2) -- (v3) -- (v4) -- (v5) -- (v1)
          (v6) -- (v7) -- (v8) -- (v9) -- (v10) -- (v6)
          (v5) -- (v6)
          (v4) -- (v8);
	\end{tikzpicture}
	\label{fig_ednm_counterexample}
    \caption{A 2-vertex-connected graph of order 10 and oriented m-eternal dominating number 6}
\end{figure}

We now show that $\oednm$ can be bounded by a natural parameter.

\begin{theorem}
Let $G = (V, E)$ be a non directed graph.
Then, $\oednm(G) \leq \oscdd(G) + 1 \leq \gamma_{2,2}(G) + 1$.
\end{theorem}

\begin{proof}
The first inequality is a consequence of Theorem \ref{ednm_scdd}.
Let $S$ be a 2-dominating 2-edge-connected set of $G$.
We construct an orientation $H$ of $G$ as follows:
we orientate the edges in $S$ such that $H[S]$ is strongly connected. This is possible thanks to Theorem \ref{robbins_theorem} since $G[S]$ is 2-edge-connected.
Then, for every vertex $v \in V \setminus S$, we choose two distinct vertices $v_1, v_2$ in $S$ 
such that $vv_1$ and $vv_2$ are edges of $G$
and orientate $vv_1$ from $v$ to $v_1$ and $vv_2$ from $v_2$ to $v$.
Thus, $S$ is a strongly connected dominating-dominated set of $H$.
\end{proof}

%This bound is tight as we show is the following proposition.
%\begin{proposition}
%For any $k > 0$, there is an infinity of graphs $G$ such that
%$\oednm(G) = \gamma_{2,2}(G) + 1 = k+1$.
%\end{proposition}

%\begin{proof}
%Let $k$ and $l$ be two integers such that $0 < k \leq l$.
%We construct the graph $G = (V, E)$ as follows.
%$V = \{ v_{i,j} : i \in [k], j \in [l] \}$
%and $E = \{ v_{i,1} v_{i+1,j} : i \in [k-1], j \in [l] \} \cup
%        \{ v_{k,1} v_{1,j} : j \in [l] \}$.  
%It is easily seen that $\{ v_{i,1} : i \in [k] \}$ is 2-edge-connected 2-dominating set of $G$.
%Thus $\gamma_{2,2}(G) \leq k$.
%Let us prove that $\oednm(G) \geq k + 1$.
%Let $H$ be a strongly connected orientation of $G$.
%Let us prove that the only dominating set of $G$ with $k$ elements is the set $\{ v_{i,1} : i \in [k] \}$.
%First notice that $N^+(v_{i,j}) = 1$ if $j > 1$.
%Thus
%$$\sum_{i=1}^k |N^+(v_{i,1})| + \sum_{i=1}^k \sum_{j=2}^l N^+(v_{i,j}) = |E|$$
%$$\sum_{i=1}^k |N^+[v_{i,1}]| = |E| + k - k(l-1) = kl = |V| + k$$

% $$|N[S]| = \sum_{v \in S} |N^+[v]| - |E(S)| < kl + k - k = |V|$$

%Thus, $S$ does not dominate $G$ (TODO a completer)
%\end{proof}

\subsection{Hardness results} % trouver un titre

Let $G$ be an undirected graph. We define $C(G)$ by starting from $G$, adding
one vertex per edge of $G$ and connecting each new vertex to the extremities of the associated edge.
See Figure \ref{ctocg_example} for an example.
This definition allows us to present the following result.

\begin{figure}
	\centering
	\begin{tikzpicture}[vertex/.style={circle,draw,fill,inner sep=2pt}]    
    
    %Graphs
    \node[vertex,label=left:$v_1$] (v1) at (0,0) {};
    \node[vertex,label=right:$v_2$] (v2) at (2,0) {};
    \node[vertex,label=left:$v_3$] (v3) at (0,2) {};
    \node[vertex,label=right:$v_4$] (v4) at (2,2) {};
    \node[vertex,label=above:$v_5$] (v5) at (1,3) {};
  
    \draw (v1) -- (v2);
    \draw (v1) -- (v3);
    \draw (v2) -- (v4);
    \draw (v3) -- (v4);
    \draw (v3) -- (v5);
    \draw (v4) -- (v5);
    
    \node[vertex,label=left:$v_1$] (w1) at (6,0) {};
    \node[vertex,label=right:$v_2$] (w2) at (8,0) {};
    \node[vertex,label=left:$v_3$] (w3) at (6,2) {};
    \node[vertex,label=right:$v_4$] (w4) at (8,2) {};
    \node[vertex,label=above:$v_5$] (w5) at (7,3) {};
    
    \node[vertex,label=left:$u_{1,2}$] (u12) at (7,-1) {};
    \node[vertex,label=left:$u_{1,3}$] (u13) at (5,1) {};
    \node[vertex,label=right:$u_{2,4}$] (u24) at (9,1) {};
    \node[vertex,label=left:$u_{3,4}$] (u34) at (7,1) {};
    \node[vertex,label=left:$u_{3,5}$] (u35) at (6,3) {};
    \node[vertex,label=right:$u_{4,5}$] (u45) at (8,3) {};
  
    \draw (w1) -- (w2);
    \draw (w1) -- (w3);
    \draw (w2) -- (w4);
    \draw (w3) -- (w4);
    \draw (w3) -- (w5);
    \draw (w4) -- (w5);
    
    \draw (w1) -- (u12) -- (w2)
          (w1) -- (u13) -- (w3)
          (w2) -- (u24) -- (w4)
          (w3) -- (u34) -- (w4)
          (w3) -- (u35) -- (w5)
          (w4) -- (u45) -- (w5)
          ;

	\end{tikzpicture}
    \caption{The house graph $G$ on the left and $C(G)$ on the right}
	\label{ctocg_example}
\end{figure}
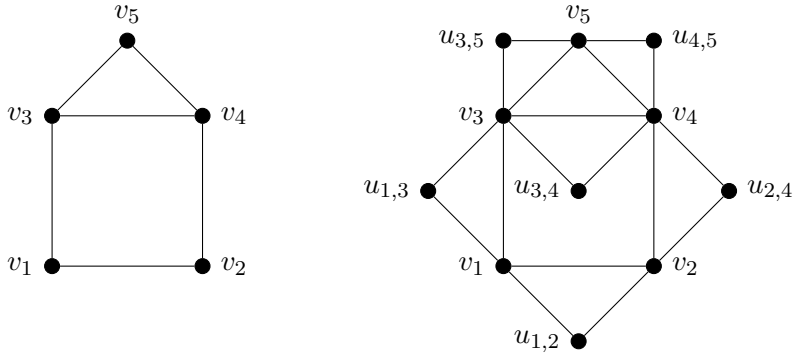

\begin{lemma}\label{GtoCG}
Let $G$ be an undirected graph with $m$ edges. Then
\begin{itemize}
\item $\oedn(C(G)) = \edn(G) + m$
\item $\oalpha(C(G)) = \alpha(G) + m.$
\end{itemize}
\end{lemma}

\begin{proof}
Let $P$ be the set of all added vertices $u_{i,j}$ in $C(G)$.

In $C(G)$, let $v_1, \ldots v_n$ be the vertices of $G$, and $u_{i,j}$, $i<j$ be the added vertex associated to the edge $\{v_i,v_j\}$.
We prove that $\oedn(C(G)) \leq \edn(G) + m$ by showing that there exists an orientation $H$ of $C(G)$ that can always be defended by $\edn(G) + m$ guards.
To construct $H$, we orientate $C(G)$ in the following way: we orientate every edge
$v_i v_j \in E$ from $v_i$ to $v_j$ iff $i<j$, then $v_j u_{i,j}$ from $v_j$ to $u_{i,j}$ and $v_i u_{i,j}$ from  $u_{i,j}$ to $v_i$.
Every triangle induced by some vertices $v_i$, $v_j$ and $u_{i,j}$ is therefore an oriented cycle. 
We will consider a strategy that preserves the following invariant: 
for each guard configuration $D$ in the strategy,
we can always partition $D$ into two sets of guards, a set $A$ of $m$ guards defending each one of the $m$ added vertices $u_{i,j}$, i.e. are either on $u_{i,j}$ or on $v_j$, and a set $B \subseteq V$ which is an eternal dominating set of $G$.

We start with the guard configuration $D_0=A_0 \cup B_0$, with $B_0$ a minimum eternal dominating set of $G$
and $A_0=P$. Clearly, this configuration verifies the invariant.
Consider a guard configuration $D = A \cup B$ that verifies the invariant and let us prove that, for every attacked vertex $r$, the defender can defend the attack and obtain a configuration $D'$ that verifies the invariant.

If $r = u_{i,j}$, then the defender moves the guard from $v_j$ to $u_{i,j}$.
If $r = v_i$, then there is a vertex $v_j \in B$ adjacent to $r$ such that
$B \cup \{ v_i \} \setminus \{ v_j \}$ is an eternal dominating set of $G$.
We have two possible cases that depend on the orientation of $\{ v_i, v_j\}$.
If $(v_j, v_i) \in E(H)$, then the defender moves the guard from $v_j$ to $v_i$.
If $(v_i, v_j) \in E(H)$, then the defender moves the guard from $u_{i,j}$ to $v_i$. 
By choosing $B' = B \cup \{ v_i \} \setminus \{ v_j \}$
and $A' = A \cup \{ v_j \} \setminus \{ u_{i,j} \}$, it is easily seen that $D' = A' \cup B'$
satisfies the invariant.

We now prove that $\oedn(C(G)) \geq \edn(G) + m$ by showing that there exists an eternal dominating set $X$ of $G$ such that $|X|=\oedn(C(G)) - m$.
Let $H$ be an orientation of $C(G)$ such that $\edn(H)=\oedn(C(G))$.
We call clean configuration of $H$ any eternal dominating set $S$ of $H$ such that $P\subseteq S$.
In $H$, the $\oedn(C(G))$ guards can always be brought to a clean configuration. Indeed, the attacker can successively attack every vertex of $P$ and be sure that they will therefore all be occupied. We prove that if a set $S$ is a clean configuration of $H$, then $S\cap V$
is an eternal dominating set of $G$. 

We first prove that $S\cap V$ is a dominating set of $G$. 
Indeed, $S$ is a dominating set of $H$. Any vertex of $G$ dominated by $v\in S\cap V$ in
$H$ is still dominated by it in $G$, and every vertex $u_{i,j}\in P$ can only dominate $v_i$ or $v_j$. If $u_{i,j}$ is the only vertex dominating $v_i$ in $S$, then if $v_i$ is attacked, a guard has to move from $u_{i,j}$ to $v_i$, which means that $u_{i,j}$ has to be dominated by $v_j$, so that $v_j\in S$. Since $v_i v_j \in E$, $v_j$ dominates $v_i$ in $G$, so that every vertex of $G$ is dominated by a vertex of $S\cap V$ in $G$.

Then, for every attack $r \in V \setminus S$, by attacking $r$ in $H$, then $u_{i,j}$ if $u_{i,j}$ became unoccupied, we obtain a new clean configuration $S'$ of $H$ so that $r \in S'\cap V$. Indeed, if a guard moves from $v\in V$ to $r$ in $H$, we obtain a clean configuration $S'=S \setminus \{v\} \cup \{r\}$, and the guard can do the same in $G$. If a guard moved from $u_{i,j}$ to $r=v_i$ in $H$, then $v_j$ is the only vertex dominating $u_{i,j}$, so that when $u_{i,j}$ is attacked, a guard must move from $v_j$ to $u_{i,j}$. We therefore obtain a clean configuration $S'=S\setminus \{v_j\} \cup \{v_i\}$, and a guard can directly move from $v_j$ to $r=v_i$ in $G$ to obtain the configuration $S' \cap V$.

Therefore, if we take a clean configuration $S$ of $H$ such that $|S|=\oedn(C(G))$, we have that $S\cap V$ is an eternal dominating set of $G$, and since $|S\cap V|=\oedn(C(G)) - m$, $\oedn(C(G)) \geq \edn(G) + m$.

We now prove that $\alpha(G) \geq \oalpha(C(G)) - m$. Let $H$ be an orientation of $C(G)$ that minimizes $\alpha(H)$. We can assume, without loss of generality, that all triplets $\{u_{i,j},v_i,v_j\}$ in $H$ induce an oriented triangle.
Indeed, consider an induced acyclic subgraph $H[S]$ in such an orientation.
Then, changing the orientation of some edges $\{ u_{i,j}, v_i \}$ or $\{ u_{i,j}, v_j \}$
does not create new minimal oriented cycles in $H$. Thus, $H[S]$ stays acyclic, and $\alpha(H)$ cannot decrease.
Therefore, at most two of the three vertices of each triplet $\{v_i,v_j,u_{i,j}\}$ belong to $S$.
We can assume, without loss of generality, $P \subseteq S$. Indeed, if $v_i,v_j\in S$, $H[S\setminus \{v_i\}\cup \{u_{i,j}\}]$ is also acyclic, and if $v_i \notin S$, $H[S\cup \{u_{i,j}\}]$ is also acyclic. Thus, $|S\cap V|=\oalpha(C(G)) - m$. Moreover, since for all $v_i v_j \in E$, $u_{i,j}\in S$, we cannot have  $v_i,v_j \in S$, and $|S\cap V|$ is an independent set of $G$. Therefore, $\alpha(G) \geq \oalpha(C(G)) - m$.

To see that $\oalpha(C(G)) \geq \alpha(G) + m$, consider a maximal independent set $S$ of $G$. Take $S'=S\cup P$. We obviously have $|S'|=\alpha(G) + m$ and $C(G)[S']$ is acyclic. Indeed, there are no cycles in $C(G)[S]=G[S]$, and an added vertex $u_{i,j}$ of $P$ cannot create any cycle since it is only linked to $v_i$ and $v_j$, and $v_i$ or $v_j$ is not in $S$. Thus, for any orientation $H$ of $C(G)$, $H[S']$ is acyclic.
\end{proof}

Notice that Lemma \ref{GtoCG} is not true if we replace $\oedn$ with $\oednm$.
Indeed, the inequality $\oednm(C(G)) \leq \ednm(G) + m$ remains true but the inequality $\oednm(C(G)) \geq \ednm(G) + m$ is not necessarily true.
For example, if we consider $G$ as the path graph $P_3$, we let the reader verify that $\ednm(G) = 2$ and $\oednm(C(G)) = 3 < \ednm(G) + m$.

The consequences of Lemma \ref{GtoCG} are particularly interesting, leading to complexity results.
The first consequence is about the (co)NP-hardness of computing $\oedn(G)$.
To our knowledge, there is no known hardness result about the complexity of computing $\edn(G)$. However, given a graph $G$ and a set $S$, deciding
whether $S$ is an eternal dominating set of $G$ is a $\Pi_2^P$-hard problem \cite{klostermeyer2007complexity}.

We will first prove that deciding, given a graph $G$ and an integer $k$, whether $\edn(G) \leq k$ is coNP-hard.
We use a reformulation of a theorem which says that $\alpha(G)$ is hard to approximate with a polynomial ratio.

\begin{theorem}\cite{zuckerman2006linear}\label{zucker}
Let $\epsilon > 0$ and $\Gamma$ be a problem with a graph $G$ and an integer $k > 0$ as input and such that:
\begin{itemize}
\item every instance $(G, k)$ with $\alpha(G) < k$ is negative;
\item every instance $(G, k)$ with $\alpha(G) \geq k^{1/\epsilon}$ is positive.
\end{itemize}

Then, $\Gamma$ is NP-hard.
\end{theorem}

\begin{theorem}\label{edn-hard}
Given a (non directed) graph $G$ and an integer $k$, deciding whether $\edn(G) \leq k$ is coNP-hard.
\end{theorem}

\begin{proof}
We use Theorem \ref{zucker} and choose $\epsilon = \frac{1}{2}$.
We consider the problem $\Gamma$: given $G, k$, do we have $\edn(G) > \binom{k+1}{2}$ ?

Clearly, there is a polynomial reduction from the complement of $\Gamma$ to the stated problem.
Thus, it suffices to prove that $\Gamma$ satisfies the conditions of Theorem \ref{zucker}. 
If $\alpha(G) < k$, then $\edn(G) \leq \binom{\alpha(G)+1}{2} < \binom{k+1}{2}$ (by Theorem \ref{eternal-seq-ineqs}), and thus $(G, k)$ is a negative instance of $\Gamma$.
If $\alpha(G) \geq k^2$, then $\edn(G) \geq k^2 > \binom{k+1}{2}$ (by Theorem \ref{eternal-seq-ineqs}).
Thus $(G, k)$ is a positive instance of $\Gamma$.
\end{proof}

From Lemma \ref{GtoCG} and Theorem \ref{edn-hard}, we obtain:

\begin{corollary}\label{oedn-coNP}
Deciding whether $\oedn(G) \leq k$ is coNP-hard.
\end{corollary}

Since deciding whether $\alpha(G) \geq k$ is NP-hard, we also obtain:

\begin{corollary}
Deciding whether $\oalpha(G) \geq k$ is NP-hard.
\end{corollary}

We believe that these lower bounds are loose and these two problems are $\Pi_2^P$-hard.

Klostermeyer and MacGillivray \cite{kloster2009} proved that there can be an arbitrary gap between $\alpha$ and $\edn$.
As consequence of Lemma \ref{GtoCG}, we show the same result between $\oalpha$ and $\oedn$.

\begin{corollary}
For every integer $k > 0$, there exists a graph $G$ such that $\oedn(G) \geq \oalpha(G) + k$.
\end{corollary}

\subsection{Results on some classes of graphs}

We are now interested in the value of $\oedn$ and  $\oednm$ for particular classes of graphs.

\subsubsection{Cycles and forests}
The case of cycle is quite straightforward for both parameters.

\begin{theorem}\label{edn_cycle}
$\oedn(C_n) = n-1$ and $\oednm(C_n) = \left \lceil \frac{n}{2} \right \rceil$ for every $n \geq 3$.
\label{cycles}
\end{theorem}

\begin{proof}
By Corollary \ref{eternal_acyclic}, for any acyclic orientation $H$ of $C_n$, we have $\edn(H) = \ednm(H) = n$.
Now consider the cyclic orientation $H$ of $C_n$.

Since $\alpha(H)=n-1$, we have $\oedn(C_n) \geq n-1$.

To see that $\edn(H)\leq n-1$, consider the following strategy with $n-1$ guards: every time a vertex is attacked, the guard on its unique incoming neighbor moves to it. Since only the attacked vertex is unoccupied, we know that the neighbor is occupied so that this defense is always possible, and leads to the exact same configuration that we were in.

Suppose that $\ednm(H) < \ffceil{\frac{n}{2}}$.
 Let $D \subseteq V(H)$ be a dominating set of size $\ednm(H)$.
There exist two vertices $u,v \in V(H) \setminus D_i $ such that $(u,v) \in E(H)$.
Since $u$ is the only vertex which dominates $v$ in $H$ and $u \notin D_i$, no vertex in $D$ dominates $v$. Thus, $D$ is not a dominating set of $H$.
Consequently, $\oednm(C_n) \geq \left \lceil \frac{n}{2} \right \rceil$.

To see that $\ednm(H) \leq \left \lceil \frac{n}{2} \right \rceil$, consider the following strategy with $\lceil \frac{n}{2} \rceil$ guards: place one guard every two vertices on $C_n$, with eventually two successive vertices if $n$ is odd.
Then, every time a vertex is attacked, move every guard to the unique outgoing vertex of their current vertex. Since the attacked vertex is unoccupied, we know that its incoming neighbor is occupied so that this defense is always possible, and leads to the exact same configuration that we were in.
\end{proof}

We now consider forests. Since they are acyclic for any orientation, we obtain the following result:

\begin{theorem}
Let $G$ be a graph with order $n$. Then, $\oedn(G) = n$ iff $\oednm(G) = n$ iff $G$ is a forest.
\end{theorem}

\begin{proof}
If $G$ is a forest, then every orientation of $G$ is acyclic and thus $\oedn(G) = \oednm(G) = n$ by Corollary \ref{eternal_acyclic}.
If $G$ is not a forest, then $G$ admits a cycle $C$ of $k$ vertices.
Consider an orientation $H$ of $G$ where the edges $C$ form an oriented cycle.
One can protect $C$ with at most $k-1$ guards and $G-C$ can be protected by at most $n-k$ guards.
Thus, $\ednm(H) \leq \edn(H) \leq n-1$.
\end{proof}

\subsubsection{Complete graphs and graphs of oriented m-eternal number 2}
We will now characterize the graphs $G$ with $\oednm(G) = 2$.
Notice that only graphs with one vertex satisfy $\oednm(G) = 1$.

\begin{lemma}Let $G = (V, E)$ be a graph of order $n \geq 2$. Then, $\oednm(G)\geq 2$.
\end{lemma}
\begin{proof}
Let $H = (V, E')$ be an orientation of $G$, and suppose that $\ednm(H)=1$. Let $u,v \in V$. Consider that the guard is on the vertex $u$ and the attack on the vertex $v$.
Thus, necessarily, $(u, v) \in E'$.
If the next attack is on $u$, the defender cannot answer since there is no edge $(v, u)$ in $E'$. We obtain a contradiction.
\end{proof}

We now prove a simple but essential lemma.

\begin{lemma}\label{nonedgelemma}
Let $G = (V, E)$ be a digraph of order $n \geq 3$ such that $\ednm(G)=2$. Let $u,v \in V$. If neither $(u,v)$ nor $(v,u)$ are edges of $G$, then $\{u,v\}$ is a dominating set of $G$.
Furthermore, the attacker can force the defender to put its two guards on $\{u, v\}$.
\end{lemma}

\begin{proof}
Let $D$ be an eternal dominating set of $G$ with $|D| = 2$.
Consider an attack on $u$ and $D'$ the answer of the defender.
We have $u \in D'$. Now, consider an attack on $v$.
Since $(u,v) \notin E$, the guard on $u$ cannot move to $v$, so that the second guard must do it. If the defender moves $u$ to a neighbor $u^+$,
then $\{v, u^+\}$ is not a dominating set since none of $u^+$ and $v$ dominates $u$.
Therefore the answer of the defender to the attack on $u$ is necessarily $\{u, v\}$. Consequently $\{u, v\}$ is a (eternal) dominating set.
\end{proof}

We can now characterize the graphs with $\oednm(G) = 2$.

\begin{theorem}\label{oednequals2}
Let $G$ be a graph of order $n \geq 3$. Then, $\oednm(G)=2$ iff either:
\begin{itemize}
\item $n=2k$ and $G$ is a complete graph from which at most $k$ disjoint edges are removed
\item $n = 2k+1$ and $G$ is a complete graph from which at most $k-1$ disjoint edges  are removed.
\end{itemize}
\end{theorem}

\begin{proof}
We first prove that if there exists three distinct vertices $u,v,w \in V$ such that $uv,uw \notin E$, then $\oednm(G)>2$.
Suppose that $\oednm(G)=2$, and let $H$ be an orientation of $G$ such that $\ednm(H)=2$. Then, by Lemma \ref{nonedgelemma}, $\{u,v\}$ is a dominating set of $H$. Since $uw \notin E$, $(v,w) \in E(H)$. But, by Lemma \ref{nonedgelemma}, $\{u,w\}$ also is a dominating set of $H$, so that $(w,v) \in E(H)$. We thus obtain a contradiction. By absurd, $\oednm(G)>2$.
Consequently, for $n=2k$ or $n=2k+1$, if $\oednm(G)=2$ then $G$ is a complete graph from which are removed at most $k$ disjoint edges.

Now take $n=2k+1$, and let $G = (V, E)$ be the complete graph from which exactly $k$ disjoint edges are removed. Let $x_1,\ldots,x_k,y_1,\ldots,y_k,z$ be the vertices in $V$ and assume that the non edges of $G$ are the pairs $x_i y_i$ for $i \in [k]$.
We show that $\oednm(G)>2$. Suppose that $\oednm(G)=2$, and let $H$ be an orientation of $G$ such that $\ednm(H)=2$.

\begin{claim}\label{lemmaclaim}
For all $i \in [k]$, $\{x_i,y_i\}$ dominates $H$.
\end{claim}
Indeed, this is a direct application of Lemma \ref{nonedgelemma}.

\begin{claim}\label{xj}
For every vertex  $v\in V \setminus \{z\}$ and for every index $i \in [k]$, $v$ dominates $x_i$ or $y_i$, but not both.
\end{claim}
Indeed, we can assume, without loss of generality, that $v=x_j$. Now, if $x_j$ dominates $x_i$ and $y_i$, $\{x_i,y_i\}$ does not dominate $x_j$, which contradicts Claim \ref{lemmaclaim}. If $v$ dominates none, by Claim \ref{lemmaclaim}, $y_j$ dominates both, so that $\{x_i,y_i\}$ does not dominate $y_j$, which contradicts Claim \ref{lemmaclaim}.

\begin{claim}\label{z}
for all $i \in [k]$, $x_i$ or $y_i$ dominates $z$, but not both.
\end{claim}
Indeed, by Claim \ref{lemmaclaim}, $x_i$ or $y_i$ dominates $z$. Since the attacker can always attack $z$, there exists $v\in V \setminus \{z\}$ such that $\{z,v\}$ dominates $H$. Now, if both $x_i$ and $y_i$ dominate $z$, $v$ dominates $x_i$ and $y_i$, which contradicts Claim \ref{xj}.

\begin{claim}\label{vdominatesz}
If $\{z,v\}$ dominates $H$, $v$ dominates $z$.
\end{claim}
Indeed, we can assume, without loss of generality, that $v=x_i$. If $z$ dominates $x_i$, then, by Claim \ref{z}, $z$ does not dominate $y_i$, but $x_i$ cannot dominate $y_i$ either, so that we obtain a contradiction.

\begin{claim}\label{visunique}
There exists a unique vertex $v$ such that $\{z,v\}$ dominates $H$.
\end{claim}
Indeed, suppose there exist $v_1,v_2 \in V \setminus \{z\}$ such that $\{z,v_1\}$ and $\{z,v_2\}$ dominate $H$. By Claim \ref{vdominatesz}, $v_1$ and $v_2$ dominate $z$. But, since $\{z,v_1\}$ dominates $v_2$, $v_1$ dominates $v_2$. Similarly, $v_2$ dominates $v_1$. We thus obtain a contradiction.

We can now assume without loss of generality that $\{z,x_i\}$ dominates $H$. By Claim \ref{vdominatesz}, $x_i$ dominates $z$, and by Claim \ref{z}, $z$ dominates $y_i$. By Lemma \ref{nonedgelemma}, the attacker can force the guards to be on $\{x_i,y_i\}$. Then, if he attacks $z$, only the guard on $x_i$ can go to $z$. The guard on $y_i$ cannot move to $x_i$, so that there exists a vertex $v$ in $V(G)\setminus \{z,x_i\}$ such that $\{z,v\}$ dominates $H$. This contradicts Claim \ref{visunique}. Thus, by absurd, $\oednm(G)>2$. Therefore, if $\oednm(G)=2$ then $n=2k$ and $G$ is the complete graph from which are removed at most $k$ disjoint edges, or $n = 2k+1$ and $G$ is the complete graph from which at most $k-1$ disjoint edges are removed.

We now prove that if $n=2k$ and $G$ is the complete graph from which at most $k$ disjoint edges are removed, then $\oednm(G)=2$. 
Without loss of generality, we assume that exactly $k$ disjoint edges are removed.
Let $V(G)=\{v_0,v_1,\ldots,v_{2k-1}\}$ and assume that the non edges of $G$
are the pairs $x_i x_{i+k}$ for $i \in [0,k-1]$.
We construct an orientation $H$ of $G$ in the following way: for all $i, j$,  $i>j$, $(v_i,v_j)\in E(H)$ iff $i-j<k$, and $(v_j,v_i)\in E(H)$ iff $i-j>k$. We will consider a strategy for the guards that preserves the following invariant: for each guard configuration $D$ in the strategy, if $D=\{v_i,v_j\}$ then $|i - j| = k$.
Notice that, by construction of $H$, every guard configuration that satisfies the invariant is a dominating set of $H$.
We start with $D_0=\{v_0,v_k\}$. Clearly, this configuration verifies the invariant.
Consider a guard configuration $D=\{v_i,v_j\}$ that verifies the invariant and let us prove that, for every attacked vertex $v_r$, the defender can defend the attack and obtain a configuration $D'$ that verifies the invariant.
Without loss of generality, we can assume that $j=i+k$.
If $i < r < j$, then we move a guard from $v_i$ to $v_r$ and the other guard from $v_j$ to $v_{r'}$
with $r' = r + k \mod 2k$.
Otherwise, we move a guard from $v_j$ to $v_r$ and the other guard from $v_i$ to $v_{r'}$.
In both cases, we obtain the configuration $\{ v_r, v_{r'} \}$ that verifies the invariant.

Assume now that $n=2k+1$ and consider $G$ to be the complete graph of order $n$ from which at most $k-1$ disjoint edges are removed.
Without loss of generality, we assume that exactly $k-1$ disjoint edges are removed.
Let us prove that $\oednm(G) = 2$.
$G$ has exactly $3$ universal vertices.
Let $z$ be one of them and let $v_0, \ldots v_{2k-1}$ be the vertices in $V \setminus \{z\}$.
Without loss of generality, we assume that $v_0$ and $v_k$ are the two other universal vertices and $v_i v_{i+k}$ is not an edge of $E$ for every $i \in [1,k-1]$.

We create an orientation $H$ of $G$ as follows.
We orientate the edges of $G[V \setminus \{ z \}]$ except $v_0 v_k$ identically to the orientation of $G$ in the even case.
Then, we orientate the edges incident to $z$ except for $z v_0$ and $zv_k$ such that $z$ has the same neighborhood as $v_0$.
Finally, we orientate the three remaining edges such that $(v_0, v_k), (v_k, z), (z, v_0) \in E(H)$.

We will give a strategy that preserves the following invariant: $D=\{v_i,v_j\}$ with $j - i = k$ or $D = \{ z, v_k \}$.
There are several cases:

case 1: if $D = \{ v_i, v_{i+k} \}$ and the attacker chooses a vertex $v_j$, then the defender plays as in the even case.

case 2: if $D = \{ v_i, v_{i+k} \}$ with $i \in [1, k-1]$ and the attacker plays in $z$,
then, the defender plays as in the even case but by replacing $v_0$ with $z$.

case 3: if $D = \{v_0, v_k \}$ and the attacker plays in $z$.
Then, the defender moves the guard on $v_0$ to $v_k$ and the guard on $v_k$ to $z$.

case 4: if $D = \{ z, v_k \}$ and the attacker plays in $v_j$ with $j > 0$, then the defender plays as in the even case but by replacing $v_0$ with $z$.

case 5: if $D = \{ z, v_k \}$ and the attacker plays in $v_0$, the defender moves the guard on $z$ to $v_0$. 
\end{proof}

The class of graphs $G$ with $\oednm(G) = 3$ seems hard to characterize.
We now consider complete graphs.
Surprisingly,
the exact value of  $\oedn$ for complete graphs seems hard to find.
However, we can obtain lower and upper bounds using a result
from Erdös and Moser concerning $\oalpha$.

\begin{theorem}\cite{erdos1964representation}
For every $n > 0$,
$\fffloor{\log_2 n} + 1 \leq \oalpha(K_n) \leq 2 \fffloor{\log_2 n} +1$
\end{theorem}

By combining this theorem with Theorem \ref{eternal-seq-ineqs-digraph}, we obtain:

\begin{corollary}\label{edn_bounds_on_clique}
For every integer $n > 0$,
$\lfloor \log_2 n \rfloor + 1 \leq \oedn(K_n) \leq \binom{2 \lfloor \log_2 n \rfloor +2}{2}.$
\end{corollary}

\subsubsection{Complete bipartite graphs}
The case of complete bipartite graphs, on the other hand, has been fully covered for both parameters.

\begin{theorem}
$\oedn(K_{n,m}) = \max\{n, m\} + 1$ for every $n, m \geq 1$.
\end{theorem}
\begin{proof}
Denote by $A$ and $B$ the two parts of $K_{n,m}$.
First, we prove that $\oedn(K_{n,m}) \geq \max\{n, m\} + 1$.
Without loss of generality, we assume that $|A| \geq |B|$. 
Let $T$ be the induced subgraph $G[A \cup \{v\}]$ where $v$ is a vertex of $B$.
It is easily seen that $T$ is a tree. Thus, any orientation of $T$ is acyclic.
Consequently, $\oalpha(K_{n,m}) \geq |A| + 1$.
By Proposition \ref{eternal_alpha_olpha}, we obtain the desired inequality.
We will now prove that $\oedn(K_{n,m}) \leq \max\{n, m\} + 1$.
By Lemma \ref{induced_edn}, we assume without loss of generality that $n = m$.
Let $M$ be a perfect matching of $K_{n,n}$.
We construct an orientation $H$ of $K_{n, n}$ as follows:
let $u \in A$ and $v \in B$. If $\{u, v\} \in M$ then $(u, v) \in E(H)$.
Otherwise $(v, u) \in E(H)$.
We start by putting a guard on every vertex of $A$ and one guard on an arbitrary vertex of $B$.
In the strategy, we preserve the following invariant:
there is at least one guard in every edge of the matching $M$ and exactly one edge $e^*$ of $M$ has a guard on its two extremities. We denote by $v^*$ the extremity of $e^*$ in $B$.
Suppose that a vertex $v$ of $B$ is attacked. Let $u$ be the vertex such that $\{u, v\} \in M$.
Then, we move the guard on $u$ to $v$.
Suppose now that a vertex $v$ of $A$ is attacked. Then we move the guard on $v^*$ to $v$.
It is easily seen that the invariant is preserved.
\end{proof}

\begin{theorem}
$\oednm(K_{2,2}) = 2$\\
$\oednm(K_{2,3}) = \oednm(K_{3,3}) = 3$\\
$\oednm(K_{n,m}) = 4$ for every $n \geq 2$ and $m \geq 4$.
\end{theorem}
\begin{proof}
$K_{2,2}$ is isomorphic to $C_4$ so, by Theorem \ref{cycles}, $\oednm(K_{2,2})=2$.

It is easily seen that $K_{2,3}$ and $K_{3,3}$ don't satisfy the conditions of Theorem \ref{oednequals2}.
Thus $\oednm(K_{2,3}) \geq 3$ and $\oedn(K_{3,3}) \geq 3$.

%Suppose that $\oednm(K_{2,3})=2$. Denote by $A$ and $B$ the two parts of $K_{2,3}$, $A=\{a_1,a_2,a_3\}$, $B=\{b_1,b_2\}$. Then, by Lemma \ref{nonedgelemma}, for any orientation $H$ of $K_{2,3}$, $\{a_1,a_2\}$ is a dominating set of $H$. But none of $a_1$ or $a_2$ dominates $a_3$, as $\{a_1,a_3\},\{a_2,a_3\} \notin E(K_{2,3})$. 
%Therefore, by absurd, $\oednm(K_{2,3})\geq 3$. Similarly, $\oednm(K_{3,3})\geq 3$.

To show that $\oednm(K_{2,3})\leq 3$, we consider the following defense : we make one guard stay on $a_3$, that will therefore always be defended. The subgraph induced by the rest of the vertices $\{a_1,a_2,b_1,b_2\}$ is isomorphic to $K_{2,2}$, and can therefore be defended by the two guards left.

As $C_6$ is a spanning subgraph of $\oednm(K_{3,3})$, and since, by Theorem \ref{cycles}, $\oednm(C_6) = 3$, we also have $\oednm(K_{3,3})\leq 3$.

We now consider the case where $n\geq 2$ and $m \geq 4$. We first show that $\oednm(K_{n,m}) \geq 4$. Suppose that $\oednm(K_{n,m}) = 3$, so that there exists an orientation $H$ of $K_{n,m}$ for which any attack can be defended by three guards. Denote by $A$ and $B$ the two parts of $K_{n,m}$, $|A|=m$, $|B|=n$. 

Consider a guard configuration $D_i$ where $|D_i\cap A|=2$ and $|D_i\cap B|=1$. Let $u,v,s,t \in A$ and $w \in B$.
Without loss of generality, we suppose that $D_i=\{u,v,w\}$ and $r_{i+1} = s$ is the attacked vertex.
We denote by $D_{i+1}$ a guard configuration reachable from $D_i$ and that defends $s$.
Only the guard on $w$ can defend $s$, so that in $H$, $w$ dominates $s$, and the guard on $w$ goes to $s$. Now the other two guards can either stay on $u$ and $v$, or go to one of their outgoing neighbors, respectively $u^+$ and $v^+$. We can therefore have four different configurations for $D_{i+1}$: $\{u,v,s\}$, $\{u^+,v,s\}$, $\{u,v^+,s\}$ or $\{u^+,v^+,s\}$. 

In the first configuration, no vertex dominates $t$, in the second, none dominates $u$ and in the third one, none dominates $v$. Therefore, we necessarily have $D_{i+1}=\{u^+,v^+,s\}$. But, since $u^+,v^+,w \in B$ and $w$ dominates $s$, no vertex dominates $w$. Thus, we have either $u^+=w$ or $v^+=w$. We can consider, without loss of generality, that $u^+=w$. Then, $u$ dominates $w$ and since $\{w,v^+,s\}$ is a dominating set, with $\{w,v^+\} \in B$ and $s \in A$, $w$ or $v^+$ must dominate $u$ and $v$, so $v^+$ dominates $u$ and $w$ dominates $v$.

Suppose now that $v^+$ dominates $s$, and the next attacked vertex $r_{i+2}$ is $u$. Only $v^+$ can go to $u$, so that we can once again obtain four different configurations for $D_{i+2}$ : $\{w,u,s\}$, $\{w^+,u,s\}$, $\{w,u,s^+\}$ or $\{w^+,u,s^+\}$. In the first and third configurations, no vertex dominates $v^+$, in the second, none dominates either $v$ or $t$. The only possibility for the last one to dominate $s$ is if $w^+=s$, but then $\{s,u,s^+\}$ would not dominate $v^+$. Therefore, we obtain a contradiction, and by absurd, $s$ dominates $v^+$.

Since $\{u,v,w\}$ dominates $H$, with $\{u,v\} \in A$ and $w \in B$, $w$ dominates $A \setminus \{u,v\}$, and since $w$ dominates $v$, only $u$ dominates $w$. Thus, for any $k$, $u \in D_k$ or $w \in D_k$. Since $\{w,v^+,s\}$ is a dominating set, with $\{w,v^+\} \in B$ and $s \in A$, $s$ dominates $B\setminus \{w,v^+\}$, and since $s$ dominates $v^+$, only $w$ dominates $s$. Thus, for any $k$, $w \in D_k$ or $s \in D_k$. Similarly, since we could have had $r_{i+1}=t$, for any $k$, $w \in D_k$ or $t \in D_k$. Therefore, for any $k$, $D_k=\{u,s,t\}$, or $w \in D_k$. But, in $\{u,s,t\}$, no vertex dominates $v$. Thus, for any $k$, $w \in D_k$. But, from $D_{i+1}=\{w,v^+,s\}$ , if $r_{i+2}=v$, only $w$ dominates $v$ so the guard on $w$ has to move to $v$, and since nor $s$ nor $v^+$ dominates $w$, we necessarily have $w \notin D_{i+2}$ and obtain a contradiction.

All considered vertices and some of the arcs in $H$ are represented in Figure \ref{Kmn1}.

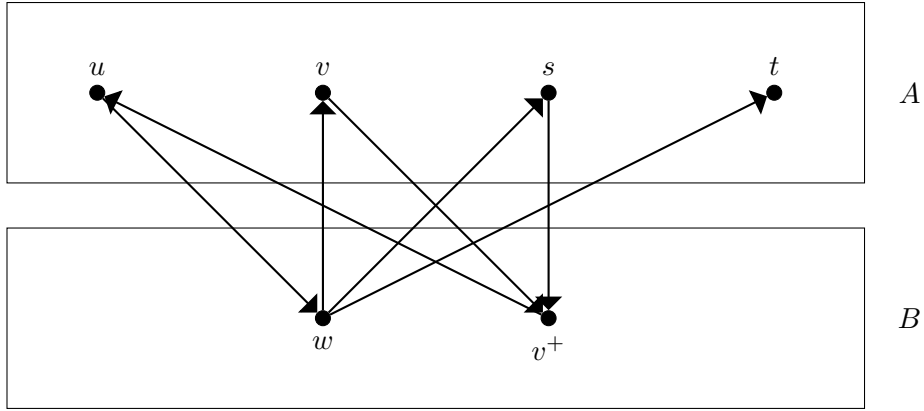
\begin{figure}
	\centering
	\begin{tikzpicture}[scale=0.3]
    
    %Graphs
    \node[draw,circle,fill,inner sep=2pt,label=below:$w$] (w) at (0,0) {};
    \node[draw,circle,fill,inner sep=2pt,label=below:$v^+$] (v+) at (10,0) {};
    \node[draw,circle,fill,inner sep=2pt,label=above:$u$] (u) at (-10,10) {};
    \node[draw,circle,fill,inner sep=2pt,label=above:$v$] (v) at (0,10) {};
    \node[draw,circle,fill,inner sep=2pt,label=above:$s$] (s) at (10,10) {};
    \node[draw,circle,fill,inner sep=2pt,label=above:$t$] (t) at (20,10) {};
    
    \draw [] (-14,6) rectangle (24,14);
    \draw [] (-14,-4) rectangle (24,4);
    \node[] at (26,10) {$A$};
    \node[] at (26,0) {$B$};
    
    \draw[-triangle 90,thick] (v+) -- (u);
    \draw[-triangle 90,thick] (u) -- (w);
    \draw[-triangle 90,thick] (w) -- (v);
    \draw[-triangle 90,thick] (v) -- (v+);
    \draw[-triangle 90,thick] (w) -- (s);
    \draw[-triangle 90,thick] (s) -- (v+);
    \draw[-triangle 90,thick] (w) -- (t);

	\end{tikzpicture}
    \caption{The vertices $u,v,s,t,w,v^+$ and some of their corresponding arcs in the orientation $H$ of $K_{m,n}$, for $D_i=\{u,v,w\}$, $r_{i+1}=s$, $D_{i+1}=\{w,v^+,s\}$ and $r_{i+2}=u$. We can see here that from $D_{i+1}=\{w,v^+,s\}$ , if we take  $r_{i+2}=v$, we have $w \notin D_{i+2}$.}
	\label{Kmn1}
\end{figure}

We now have to prove that the attacker can always bring the guards to a set $D_i$ such that $|D_i\cap A|=2$ and $|D_i\cap B|=1$.
Suppose there exists $i$ such that $|D_i\cap A|=1$ and $|D_i\cap B|=2$. Let $u,v\in B$ and $w,s,t \in A$. We take $D_i=\{u,v,w\}$, and $r_{i+1}=s$.
Only the guards on $u$ or $v$ can go to $s$. We can suppose, without loss of generality, that the guard on $u$ goes to $s$, so that $u$ dominates $s$. We can obtain four different configurations for $D_{i+1}$ : $\{s,v,w\}$, $\{s,v^+,w\}$, $\{s,v,w^+\}$ or $\{s,v^+,w^+\}$. 

In the first and last configurations, we have $|D_{i+1} \cap A|=2$ and $|D_{i+1} \cap B|=1$, so we are done. Let us suppose it is not the case.
In the second configuration, since $s,v^+,w \in A$, and $|A|=m\geq 4$, there exists a vertex in $A$ that is not dominated by any of the three vertices. Therefore, we necessarily have $D_{i+1}=\{s,v,w^+\}$. But, since $u,w^+,v \in B$ and $u$ dominates $s$, to dominate $u$, we must have $w^+=u$, $w$ dominates $u$ and $D_{i+1}=\{s,v,u\}$.

Now, similarly, if we take $r_{i+2}=t$, we must have $D_{i+2}=\{v,t,u\}$ and either $u$ or $v$ dominates $t$ and is dominated by $s$. Since $u$ dominates $s$, it must be $v$. Similarly, if we take $r_{i+3}=w$, we must have $D_{i+3}=\{v,w,u\}$ and either $u$ or $v$ dominates $w$ and is dominated by $t$. Since $v$ dominates $t$, and $w$ dominates $u$, we obtain a contradiction. Therefore, if there exists a configuration $D_i$ such that $|D_i\cap A|=1$ and $|D_i\cap B|=2$, then
there exists a sequence of attacks such that either the attacker win or we obtain a guard configuration $D_j$, $j > i$ with  $|D_j\cap A|=2$ and $|D_j\cap B|=1$.
All considered vertices and some of their arcs are represented in Figure \ref{Kmn2}.

%\vspace{1cm}
\begin{figure}
	\centering
	\begin{tikzpicture}[scale=0.3]
    
    %Graphs
    \node[draw,circle,fill,inner sep=2pt,label=below:$u$] (u) at (0,0) {};
    \node[draw,circle,fill,inner sep=2pt,label=below:$v$] (v) at (10,0) {};
    \node[draw,circle,fill,inner sep=2pt,label=above:$w$] (w) at (-5,10) {};
    \node[draw,circle,fill,inner sep=2pt,label=above:$s$] (s) at (5,10) {};
    \node[draw,circle,fill,inner sep=2pt,label=above:$t$] (t) at (15,10) {};
    
    \draw [] (-14,6) rectangle (24,14);
    \draw [] (-14,-4) rectangle (24,4);
    \node[] at (26,10) {$A$};
    \node[] at (26,0) {$B$};
    
    \draw[-triangle 90,thick] (w) -- (u);
    \draw[-triangle 90,thick] (u) -- (s);
    \draw[-triangle 90,thick] (s) -- (v);
    \draw[-triangle 90,thick] (v) -- (t);
    \draw[-triangle 90,thick] (v) -- (w);

	\end{tikzpicture}
    \caption{The vertices $u,v,w,s,t$ and some of their corresponding arcs in the orientation $H$ of $K_{m,n}$, if we suppose there is no $j$ s.t. $|D_j\cap A|=2$ and $|D_j\cap B|=1$, and have $D_i=\{u,v,w\}$, $r_{i+1}=s$, $D_{i+1}=\{s,v,u\}$, $r_{i+2}=t$ and $D_{i+2}=\{v,t,u\}$. We can see here that if we take $r_{i+3}=w$, only the guard on $v$ can move to $w$, but then the guard on $t$ cannot move to $v$.}
	\label{Kmn2}
\end{figure}
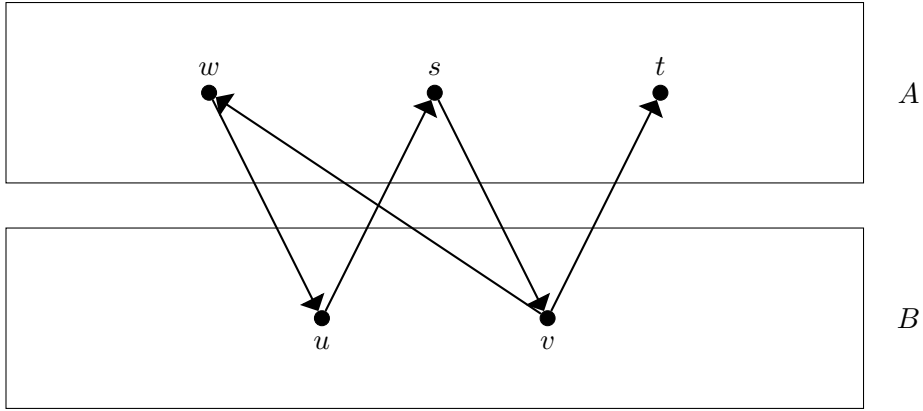

Now, consider a guard configuration $D_i$ with $|D_i\cap A|=3$. Since $|A|=m\geq 4$, there exists a vertex in $A$ that is not dominated by any of the three vertices, so we obtain a contradiction. 
If $D_i$ is a guard configuration with $|D_i \cap B|=3$ and the next attack is on a vertex in $A$, for any answer of the defender, we obtain one of the previously seen possibilities.
Therefore, for any defense, there exists $i$ such that $|D_i\cap A|=1$ and $|D_i\cap B|=2$, so that we have $\oednm(K_{n,m}) \geq 4$.

We now prove that $\oednm(K_{n,m}) \leq 4$. To see it, consider the following strategy with $4$ guards: partition $A$ into two non empty sets $A_1$ and $A_2$, and $B$ into two non empty sets $B_1$ and $B_2$. Orientate every edge of $K_{n,m}$ either from $A_1$ to $B_1$, from $B_1$ to $A_2$, from $A_2$ to $B_2$ or from $B_2$ to $A_1$. Start with a guard in each of the four sets of vertices. Exactly one guard dominates one set. Every time a vertex is attacked, move the guard who dominates its corresponding set to the vertex, and move all the three other guards in the set they dominate. This leads to the exact same configuration that we were in, so that we can always defend any attack like this.

\end{proof}

\subsubsection{Trivially perfect graphs}
Trivially perfect graphs can be characterized in various ways.
We use the following definition due to Wolk \cite{wolk_triviallly_perfect}.
A graph $G$ is trivially perfect if every connected  induced subgraph of $G$ admits a universal vertex.
For this class of graphs, we obtain the exact value of $\oednm$.

\begin{theorem}\label{eternal_trivially_perfect}
Let $G$ be a connected trivially perfect graph with at least 2 vertices and $l$ 2-vertex-connected components.
Then, we have:
\begin{itemize}
\item  if $G$ admits exactly one 2-vertex-connected component $G_i$ of size at least 3  then  $\oednm(G) = \oednm(G_i) + l - 1$ with $\oednm(G_i) \leq 3$.
\item if every 2-vertex-connected component is of size at most 3, then
      $\oednm(G) = l+1$.
\item if $G$ admits at least two 2-vertex-connected components
of size at least 4 and 3, respectively,  then   $\oednm(G) = l+2$.
\end{itemize}
\end{theorem}

We split the proof of this theorem in several lemmas.

\begin{lemma}
 $\oednm(G) \leq l+2$.
\end{lemma}

\begin{proof}
Let $x$ be a universal vertex of $G$. If $l \geq 2$, then $x$ is unique and all 2-vertex-component $G_i$ contains $x$ and another vertex $x_i$ universal for $G_i$.
We orientate $G$ as follows, for every vertex $v$ different from $x$ and from every $x_i$,
we orientate the edge $x v$  from $x$ to $v$.
For every vertex $x_i$ we orientate $x x_i$ from $x_i$ to $x$.
For every vertex $v$ in a component $G_i$ and different from $x$ and $x_i$,
we orientate the edge $x_i v$ from $v$ to $x_i$.
The other edges are oriented arbitrarily, they are not useful to the defense of $G$.

The strategy of the defender consists in permanently satisfying the following invariant:\\
1) there is a guard on $x$;\\
2) every component $G_i$ contains at least two guards (including the one on $x$);\\
3) the component $G_i$ with three guards has a guard on $x_i$.

We denote by $G^*= G_i$ the component with three guards, on the vertices $x$, $x^* = x_i$
and another vertex $v$.
Let us prove that the defender can maintain the invariant.
If the attack is on a vertex $x_i$ different from $x^*$, then the defender moves the guard on the vertex in the component $G_i$ that is not $x$ to $x_i$.
If the attack is on a vertex $v$ in $V(G_i) \setminus \{ x, x_i \}$,
then the defender moves the guard on  $x$ to $v$ , the guard on $x^*$ to $x$ and eventually 
the guard on a vertex in $V(G_i) \setminus \{ x \}$ to $x_i$ if it is not already on it.
Thus, the invariant is maintained.
\end{proof}

\begin{lemma}
if every 2-vertex-connected component is of size at most 3, then
 $\oednm(G) = l+1$.
\end{lemma}

\begin{proof}
We first prove that $\oednm(G) \leq l+1$.
We put a guard on $x$.
We orientate each component isomorphic to $K_3$ cyclically and we put a guard on the vertex which is an out-neighbor of $x$.
We orientate each component isomorphic to $P_2$ arbitrarily and we put a guard on the vertex that is not $x$.
It is easily seen that the defender can permanently defend and maintain a guard on $x$.

Now, we prove that $\oednm(G) \geq l+1$.
By Corollary \ref{eternal_2_vertex_compoenents}, we can assume that $G$ has no component isomorphic to $P_2$ since we need to permanently put a guard on the pendent vertex.
Thus, $G$ is 2-arc-connected, and we can consider a strongly connected orientation of $G$ (Theorem \ref{robbins_theorem} and Proposition \ref{eternal_strongly_connected}). If there is less than $l+1$ guards, then there is a component $G_i$ with at most one guard. Since the orientation is strongly connected, one vertex of $G_i$ is not dominated.
\end{proof}

\begin{lemma}
If  $G$ admits exactly one 2-vertex-connected $G_i$ component of size at least 3 then  $\oednm(G) = \oednm(G_i) + l - 1$.
\end{lemma}

\begin{proof}
This is a straightforward consequence of Corollary \ref{eternal_2_vertex_compoenents}.
\end{proof}

\begin{lemma}
If $G$ admits two 2-vertex-connected components
of size at least 4 and 3, respectively,  then   $\oednm(G) \geq l+2$.
\end{lemma}

\begin{proof}
Consider an orientation $H$ of $G$ and a defense of $H$ with $l+1$ guards.
Let $x$ be the universal vertex of $G$.
By Proposition \ref{eternal_strongly_connected}, we assume that $H$ is strongly connected.
Let $G_1, G_2$ be the 2-vertex-connected components of $G$ of size at least 4 and 3 respectively and $H_1, H_2$ be the corresponding subgraphs in $H$.

First, we will prove that there exists a vertex $v_1 \in H_1 \setminus \{x\}$ such that $\{x, v_1\}$ is not a dominating set of $H_1$. Assume this not the case.
We consider two cases. 1) $H_1 \setminus \{x\}$ admits no vertex with outdegree 0.
Let $v$ be a vertex of $H_1 \setminus \{x\}$ and $u$ be an out-neighbor of $v$.
Then $v$ is not dominated by $u$. Since $\{x, u\}$ dominates it, there is an edge
$(x, v)$. Thus $x$ dominates all vertices in $H_1$. That contradicts the fact that $H$ is strongly connected.
2) There is a vertex $v$ with outdegree 0 in $H_1 \setminus \{x\}$.
Thus, $v$ dominates $x$ and $x$ dominates all vertices in $H_1$ except $v$.
Let $a$ and $b$ two other vertices of $H_1$ with $(a, b) \notin E(H)$.
The attacker first chooses the vertex $a$. So the defender must moves a guard on $a$
and another guard on $x$ or $v$. Then the attacker chooses the vertex $b$.
If the defender can answer, then the second guard was necessarily on $x$.
It moves the guard on $x$ to $b$. If it moves the second guard on $v$, then $a$ is not dominated and if it moves the second guard elsewhere, $x$ is not dominated.
Notice that the guards on other components cannot participate in the defense of $H_1$.
Indeed, if such a guard  moves in $x$, then its component is not dominated.

At any time, by choosing $v_1$, since $\{x,v_1\}$ is not a dominating set of $H_1$, the attacker can force the defender to move two guards on $H_1 \setminus \{x\}$. If there are two guards on $H_1 \setminus \{x\}$, necessarily, we are in the following configuration.
Each component $H_i$ except $H_1$ contains exactly one guard on a vertex $x_i$. Since $H$ is strongly connected,
the guard on $x_i$ cannot dominate its whole component $H_i$ but it can eventually dominate 
$H_i \setminus \{ x \}$. Notice that means $x_i \neq x$ and there is no guard on $x$.
Thus, the strategy of the attacker is as follows:
first, it attacks $v_1$ so that we are in the previously described situation. Then, it chooses a vertex $y$ in $H_2$ different from $x_2$ and then it attacks $v_1$ again, which forces the defender to move two guards on $H_1 \setminus \{x\}$.
The guard on $y$ cannot come back to $x_2$ and there is no guard on $x$. Thus $H_2$ is not dominated and the attacker wins.
\end{proof}

This concludes the proof of Theorem \ref{eternal_trivially_perfect}.

By combining Theorems \ref{eternal_trivially_perfect} and \ref{oednequals2}, we obtain a full characterization of the value of $\oednm$ for trivially perfect graphs.
It is easily seen that these characterizations can be checked in linear time. So we obtain the following result.

\begin{corollary}
$\oednm$ can be computed in linear time on trivially perfect graphs.
\end{corollary} 

\subsubsection{Grids and products of graphs}
We now consider $\oedn$ on grids.
Contrary to $\edn$ (which is $\ffceil{\frac{nm}{2}}$), we think that the exact value of $\oedn$ cannot be expressed by a simple formula. We give here lower and upper bounds.

The following proposition has been verified using a computer.

\begin{proposition}\label{computer_proven}
$\oedn(P_3 \square P_3) = 7$.
\end{proposition}

The unique orientation $H$ with $\edn(H) = 7$ (up to isomorphism) is shown in Figure \ref{fig-edn-grid33}.
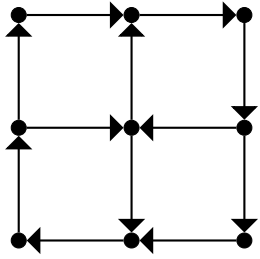
\begin{figure}
\begin{center}
\begin{tikzpicture}[vertex/.style={circle,draw,fill,inner sep=2pt}]
  \foreach \x in {0,...,2}
    \foreach \y in {0,...,2} 
       \node [vertex]  (\x\y) at (1.5*\x,1.5*\y) {};
  
   \draw[-triangle 90,thick] (00)->(01);
   \draw[-triangle 90,thick] (01)->(02);
   \draw[-triangle 90,thick] (02)->(12);
   \draw[-triangle 90,thick] (12)->(22);
   \draw[-triangle 90,thick] (22)->(21);
   \draw[-triangle 90,thick] (21)->(20);
   \draw[-triangle 90,thick] (20)->(10);
   \draw[-triangle 90,thick] (10)->(00);
   \draw[-triangle 90,thick] (01)->(11);
   \draw[-triangle 90,thick] (11)->(10);
   \draw[-triangle 90,thick] (21)->(11);
   \draw[-triangle 90,thick] (11)->(12);  
\end{tikzpicture}
\end{center}
\caption{An orientation $H$ of $P_3 \square P_3$ with $\edn(H) = 7$}
\label{fig-edn-grid33}
\end{figure}

We will show that $\frac{2nm}{3} \leq \oedn(P_n \square P_m) \leq \frac{7nm}{9} + O(n+m)$.
The next two theorems give more precise bounds.

\begin{theorem}\label{edn_grid_low}
$$\ffceil{\frac{n}{2}}m + \fffloor{\frac{m}{2}} \ffceil{\frac{n}{3}} \leq \oedn(P_n \square P_m).$$
\end{theorem}

\begin{proof}
Consider the graph $G = P_n \square P_m$ with $n$ lines and $m$ columns.
We denote by $v_{i,j}$ the vertex at line $i$ and column $j$.
To prove the lower bound, we will show that
$\alpha(H) \geq \ffceil{\frac{n}{2}}m + \fffloor{\frac{m}{2}} \ffceil{\frac{n}{3}}$ for any orientation $H$ of $G$.
We construct a set $S$ such that $H[S]$ is acyclic as follows.
First, we put odd lines vertices into  $S$.
Obviously, $H[S]$ is acyclic and $|S| = \ffceil{\frac{n}{2}}m$.
Now, we will show that one can add $\ffceil{\frac{n}{3}}$ vertices for each even line and $H[S]$ remains acyclic.
Let $i \in [n]$ be even.
If $i = n$, clearly, we can add to $S$ the vertices $v_{i,j}$
with $j$ odd without creating cycles.
Assume that $i < n$.
We split the vertices in the line $i$ into 3 types.
$A$ is the set of vertices $v_{i,j}$ where
$(v_{i-1,j},v_{i,j})$ and $(v_{i,j},v_{i+1,j})$ are arcs of $H$.
$B$ is the set of vertices $v_{i,j}$ where
$(v_{i+1,j},v_{i,j})$ and $(v_{i,j},v_{i-1,j})$ are arcs of $H$.
$C$ is the set of vertices $v_{i,j}$ that does not belong to $A$ or $B$.

There are three cases.\\
case 1: $|A| \geq \ffceil{\frac{n}{3}}$. We add all vertices of type $A$ to $S$ and we don't create cycles.\\
case 2: $|B| \geq \ffceil{\frac{n}{3}}$. Similar to case 1.\\
case 3: $|A| < \ffceil{\frac{n}{3}}$ and $|B| < \ffceil{\frac{n}{3}}$.
Without loss of generality, we assume that $|A| \geq |B|$.
We construct a set $D$ from $A \cup C$ by picking one vertex out of two in the ordered sequence of vertices of $A \cup C$.
Since $|A \cup C| \geq \ffceil{\frac{2n}{3}}$,
we have $D \geq \ffceil{\frac{n}{3}}$.
We add every vertex of $D$ in $S$.
The only manner to create cycles with elements of $A \cup C$ is to choose two consecutive elements but it is not possible by construction of $D$.
\end{proof}

\begin{theorem}\label{edn_grid_up}
For $m=3p+2x$ and $n=3q+2y$ with $p,q\in \mathbb{N}$ and $x,y \in \{0,1,2\}$, we have:
$$\oedn(P_n \square P_m) \leq 7pq + \ffceil{\frac{9p}{2}}y + \ffceil{\frac{9q}{2}}x +  3xy$$
\end{theorem}

\begin{proof}
We divide the grid into 4 parts of size $3p \times 3q$, $3p \times 2y$, $2x \times 3q$ and $2x \times 2y$ respectively.
We already know from Proposition \ref{computer_proven} that $\oedn(P_3 \square P_3)=7$ and from Theorem \ref{cycles} that $\oedn(P_2 \square P_2)=3$.
Thus, the grid of size $3p \times 3q$ can be protected with $7pq$ guards
by dividing it into squares of size $3 \times 3$.
Similarly, the grid of size $2x \times 2y$ can be protected with $3xy$ guards.
The two remaining parts can be covered by squares of size $2 \times 2$ and $1 \times 1$ and can therefore be protected with $\ffceil{\frac{9p}{2}}y$ and $\ffceil{\frac{9q}{2}}x$ guards respectively.
\end{proof}

For grids of size $2 \times n$, $3 \times n$ and $4 \times n$, the lower bound of Theorem \ref{edn_grid_low} and the upper bound of Theorem \ref{edn_grid_up} coincide and we have the exact value of $\oedn$.

\begin{corollary}
Let $n \geq 2$. Then,\\
$\oedn(P_2 \square P_n) = \ffceil{\frac{3n}{2}}$,\\
$\oedn(P_3 \square P_n) = \ffceil{\frac{7n}{3}}$,\\
$\oedn(P_4 \square P_n) = 2 \ffceil{\frac{3n}{2}}$.
\end{corollary}

We don't know the value of $\oedn(P_5 \square P_5)$.
Using Theorems \ref{edn_grid_low} and \ref{edn_grid_up},
we obtain $19 \leq \oedn(P_5 \square P_5) \leq 20$.

We now study $\oednm$ on various kinds of grids.

\begin{theorem}
For every $n \geq 2$ and $m \geq 2$, we have
$$\oednm(P_n \square P_m) \leq \ffceil{\frac{nm}{2}}.$$
\end{theorem}

\begin{proof}
If $n$ or $m$ are even, $P_n \square P_m$ admits an hamiltonian cycle and thus
$\oednm(P_n \square P_m) \leq \frac{nm}{2}$ by Theorem \ref{edn_cycle}.
Otherwise, $P_n \square P_m$ admits an hamiltonian cycle if we remove a corner vertex.
If we keep a guard on the corner vertex and defend the remaining vertices with $\frac{nm-1}{2}$ vertices, we obtain the desired bound.
\end{proof}

We don't have lower bounds except the straightforward bound $\ffceil{\frac{mn}{4}}$.  
On the other hand, the upper bound seems loose but we have verified, using a computer, that $\oednm(P_n \square P_m) = \ffceil{\frac{nm}{2}}$
for every $n$ and $m$ between 2 and 5. No counterexample has been found for other values.
We lack tools to find tight lower bounds.

We now consider upper bounds on $\oednm$ for toroidal grids, rook's graphs, toroidal kings grid and toroidal hypergrids. We present a general method based on the neighborhood-equitable coloring, a notion we introduce.

\begin{definition}
Let $k$ and $l$ be two integers and $G$ be a $(k-1)l$-regular graph.
A $(k, l)$-NE coloring of $G$ is a proper coloring $(V_i, \ldots, V_k)$ of $G$ with $k$ colors such that
for every vertex $v$ and color $i$ such that $v \notin V_i$, we have $|N(v) \cap V_i| = l$.
\end{definition}

\begin{theorem}\label{eternal_necoloring}
Let $G = (V, E)$ be a graph that admits a $(k, 2l)$-NE coloring.
Then $\oednm(G) \leq \frac{n}{k}$.
\end{theorem}

\begin{proof}
Consider a $(k, 2l)$-NE coloring $(V_1, \ldots V_{k})$ of $G$.
Let $G_{ij}$ be the subgraph of $G$ induced by $V_i \cup V_j$.
By construction, $G_{ij}$ is a $2l$-regular bipartite graph.
We orientate $G_{ij}$ such that the indegree and outdegree of every vertex is $l$.
Indeed, each component of $G_{ij}$ is eulerian. So we can orientate each component to obtain
eulerian orientations.
We do this for every distinct $i$ and $j$ in $[k]$ and we obtain an orientation $H$ of $G$.
Let us prove that $\ednm(H) \leq \frac{n}{k}$.
We initially put all guards in an arbitrary color class $V_i$.
If a vertex $v \in V_j$ is attacked, we move all guards from $V_i$ to $V_j$.
Indeed, consider the graph $B_{ij}$ with vertices $V_i \cup V_j$ and where we put an edge between two vertices $u \in V_i$ and $v\in V_j$ iff $(u, v) \in E(H)$.
$B_{ij}$ is $l$ regular by construction.
Thus, by application of Hall's marriage theorem \cite{hall1935representatives}, $B_{ij}$ admits a perfect matching between $V_i$ and $V_j$.
Consequently, there is a multimove from $V_i$ to $V_j$ in $H$.
\end{proof}

Products of graphs admit this nice property.

\begin{theorem}\label{eternal_necoloring_product}
Let $G_1$ be a graph that admits a $(k, l_1)$-NE coloring
and $G_2$ be a graph that admits a $(k, l_2)$-NE coloring.
Then, $G_1 \square G_2$ admits a $(k, l_1+l_2)$-NE coloring.
\end{theorem}

\begin{proof}
We assume that the  vertices of $G_1$ and $G_2$ are colored with integers chosen in the set $[0, k-1]$.
Let $v_1, \ldots, v_n$ be the vertices of $G_1$ and $u_1, \ldots, u_m$ be the vertices of $G_2$.
For a vertex $v_i$ of $G_1$ and a vertex $u_j$ of $G_2$, we denote by $w_{i, j}$
the vertex associated to $(v_i, u_j)$ in  $G_1 \square G_2$.
Let $p$ be the color of $v_i$ and $q$ be the color of $u_j$.
Then, we assign to $w_{i, j}$ the color $r = p+q \mod k$.
Let $r'$ be a color different from $r$.
Then $w_{i, j}$ has exactly $l_1$ (resp. $l_2$) neighbors $w_{i', j'}$ of color $r'$ with $j = j'$ (resp. $i = i'$) and thus $l1+l2$ neighbors of color $r'$.
\end{proof}

This notion of coloring has direct consequences on toroidal grids (i.e. cartesian products of two cycles) and rook's graphs (i.e. cartesian products of two complete graphs).

\begin{theorem}\label{oednm-toric-grids}
When $m$ and $n$ are both multiples of $3$, we have:
$$\oednm(C_n \square C_m) \leq \frac{nm}{3}.$$

In general, we have:
$$\oednm(C_n \square C_m) \leq \ffceil{\frac{nm}{3}} + O(n+m).$$
\end{theorem}

\begin{proof}
The first inequality is a direct consequence of Theorems \ref{eternal_necoloring}
and \ref{eternal_necoloring_product} with the fact that a cycle graph of order that is a multiple of 3 admits a $(3, 1)$-NE coloring.

%Consider first the case where $n$ and $m$ are multiple of 3.
%We affect to the vertex $v_{i, j}$ at position $(i, j)$ in the grid the color $i + j \mod 3$.
%See Figure \ref{NEcoloring-grid} for an example of coloring for a square $3 \times 3$ of the grid.
%It is easily seen that we obtain a $(3, 2)$-NE coloring.

If $n$ or $m$ is not a multiple of 3, consider the grid $C_{3n+x} \square C_{3m+y}$
with $x, y \in \{0, 1, 2\}$ and $x > 0$ or $y > 0$.
Let $H'$ be an orientation of $C_{3n} \square C_{3m}$ as described in the previous case.
We construct an orientation $H$ of $C_{3n+x} \square C_{3m+y}$ as follows.
We orientate each edge between $(i, j)$ and $(i, j+1)$ for $i \in [3n]$, $j \in [3m-1]$ and between $(i, j)$ and $(i+1, j)$ for $i \in [3n-1]$, $j \in [3m]$ in the same direction as in $H'$.
For every $i \in [3n]$, if $( v_{i, 3m}, v_{i, 1} ) \in E(H'))$, then we orientate $H$
such that $( v_{i, 3m}, v_{i, 3m+1}, \ldots, v_{i, 3m+y}, v_{i, 1})$ is an oriented path.
Otherwise, we orientate $H$ such that
$(v_{i, 1}, v_{i, 3m+y}, v_{i, 3m+y-1}, \ldots, v_{i, 3m})$
is an oriented path.
We do the same for every edge $v_{3n, j} v_{1, j}$ with $j \in [3m]$
We arbitrarily orientate the remaining edges.
An example of orientation is described in Figure \ref{NEcoloring-grid3}.

Consider the set of vertices $S$ of $H$ including a m-eternal dominating set of $H'$, as described
in the previous case, and containing every vertex $(i, j)$ with $i > 3n$ or $j > 3m$.
Then, $S$ is a m-eternal dominating set of $H$.
Indeed, we mimic the strategy of the defender for $H'$.
The only difference is when a guard in $H'$ goes from a "border" of the grid to the opposite.
For example, a guard goes from a vertex $(i, 3m)$ to the vertex $(i, 1)$.
Then, we push every guard, except for the last one, in the path
$(v_{i, 3m}, v_{i, 3m+1}, \ldots, v_{i, 3m+y}, v_{i, 1})$ to the next vertex. One can easily generalize to the other borders.
\end{proof}

%\begin{figure}
%\begin{center}
%\begin{tikzpicture}[darkstyle/.style={circle,draw,fill=gray!40,minimum size=20}]
%  \foreach \x in {1,...,3}
%    \foreach \y in {1,...,3} {
%       \pgfmathtruncatemacro{\label}{mod(\x - \y + 2, 3) + 1}
%       \node [darkstyle]  (\x\y) at (1.5*\x,1.5*\y) {\label};
%    }
%  \foreach \x in {1,...,3} {
%       \node  (0\x) at (0.5,1.5*\x) {};
%       \node  (4\x) at (1.5*4-0.5,1.5*\x) {};  
%       \node  (\x 0) at (1.5*\x, 0.5) {};
%       \node  (\x 4) at (1.5*\x, 1.5*4-0.5) {};
%  }

%  \foreach \x in {1,...,3}
%    \foreach \y [count=\yi] in {0,...,3}  
%      \draw (\x\y)--(\x\yi) (\y\x)--(\yi\x) ;

%\end{tikzpicture}
%\caption{$(2, 2)$-PE coloring of a square of a toroidal grid}
%\label{NEcoloring-grid}
%\end{center}
%\end{figure}

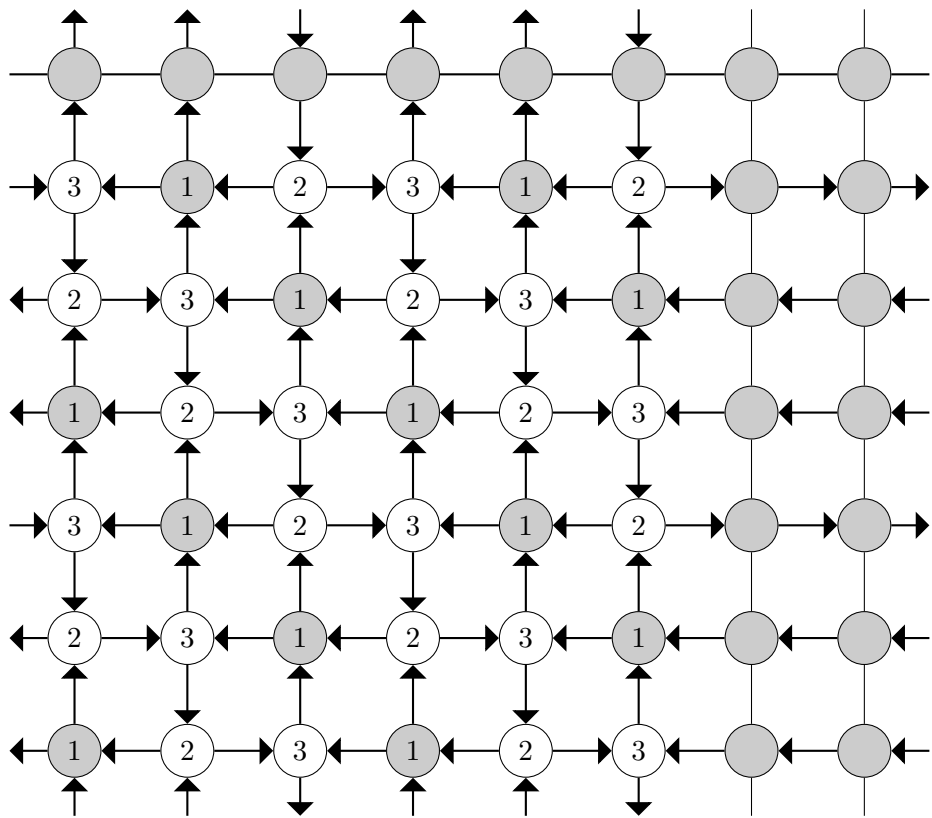
\begin{figure}
\begin{center}
\begin{tikzpicture}[vertex/.style={circle,draw,minimum size=20}]
  \foreach \x in {1,...,8}
    \foreach \y in {1,...,7} {
       \pgfmathtruncatemacro{\label}{
       \x < 7 && \y < 7 ? mod(\x + \y - 2,3)+1 : 0
       }
       \ifthenelse{\label=0}
          {\node [vertex,fill=gray!40]  (\x\y) at (1.5*\x,1.5*\y) {};}
          {\ifthenelse{\label=1}
            {\node [vertex,fill=gray!40]  (\x\y) at (1.5*\x,1.5*\y) {\label};}
           {\node [vertex]  (\x\y) at (1.5*\x,1.5*\y) {\label};}
         }
      }

  \foreach \y in {1,...,7} {
    \node (0\y) at (0.5,1.5*\y) {};
    \node (9\y) at (1.5*9-0.5,1.5*\y) {};
  }
  
  \foreach \x in {1,...,8} {
    \node (\x 0) at (1.5*\x,0.5) {};
    \node (\x 8) at (1.5*\x,1.5*8-0.5) {};
  }
  
  \foreach \x [count=\xi] in {0,...,8}
    \foreach \y  in {1,...,7} {
      \pgfmathtruncatemacro{\cond}{
       (\x > 0 && \x < 7) ? mod(\x + \y,3) == 0 : mod(6 + \y,3) == 0
       }
       \pgfmathtruncatemacro{\xi}{\x+1}
      \ifthenelse{\y < 7}{
        \ifthenelse{\cond=1}
         {\draw[-triangle 90,thick] (\x\y)--(\xi\y); }
         {\draw[-triangle 90,thick] (\xi\y)--(\x\y); }
       }
       {\draw [thick] (\x\y)--(\xi\y);}
     }
  \foreach \x  in {1,...,8}
    \foreach \y [count=\yi] in {0,...,7} {
      \pgfmathtruncatemacro{\cond}{
       (\y > 0 && \y < 7) ? mod(\x + \y,3) == 0 : mod(6 + \x,3) == 0
       }
      \ifthenelse{\x < 7}{
        \ifthenelse{\cond=1}
         {\draw[-triangle 90,thick] (\x\yi)--(\x\y); }
         {\draw[-triangle 90,thick] (\x\y)--(\x\yi); }
       }
       {\draw (\x\y)--(\x\yi);}
     }   
 % \draw[-triangle 90,thick] (77)--(78); 
  %\draw[-triangle 90,thick] (78)--(88); 
  %\draw[-triangle 90,thick] (88)--(87); 
  %\draw[-triangle 90,thick] (87)--(77); 

\end{tikzpicture}
\caption{orientation of the toroidal grid $C_8 \square C_7$}
\label{NEcoloring-grid3}
\end{center}
\end{figure}

Rook's graphs are cartesian products of two complete graphs.
They received their names from the legal moves of the rook chess piece on a chessboard.
For square rook's graphs, we obtain the exact value of $\oednm$.

\begin{theorem}\label{oednm_rook}
For every $n \geq 1$, we have
$\oednm(K_n \square K_n) = \gamma(K_n \square K_n) = n.$
\end{theorem}

\begin{proof}
It easily seen that $\gamma(K_n \square K_n) \geq n$.
Indeed, a set of size lower than $n$ does not dominate at least a line and a column. Thus, it does not dominate the vertex which is at the intersection of this line and this column.
The upper bound is a direct consequence of Theorems \ref{eternal_necoloring}
and \ref{eternal_necoloring_product} with the fact that a complete graph of order $n$ admits a $(n, 1)$-NE coloring. 
\end{proof}

Toroidal king's grids are the strong product of two cycles.
They received their names from the legal moves of the king chess piece on a (toroidal) chessboard.
For this class of graphs, we obtain the following result.

\begin{theorem}\label{oednm-toric-kinggrids}
Let $m$ and $n$ be two multiples of $5$.
Then, we have:
$$\oednm(C_n \boxtimes C_m) \leq \frac{nm}{5}.$$
%In general:
%$$\oednm(C_n \boxtimes C_m) \leq \ffceil{\frac{nm}{5}} + O(n+m).$$
\end{theorem}

\begin{proof}
We color the vertex at position $(i, j)$ with the color $i + 2j \mod 5$.
If we split the grid in squares of size $5 \times 5$,
each square can be colored as in Figure \ref{necoloring-king}.
An easy case study permits to conclude that we obtain a $(5, 2)$-NE coloring.

%If $n$ or $m$ are not multiple of 5, we proceed like in Theorem \ref{oednm-toric-grids}.
%Consider a king's grid $C_{5n+x} \boxtimes C_{5m+y}$ with $x, y \in \{0, 2, 3, 4, 6\}$
%and $H'$ be an orientation of $C_{5n} \boxtimes C_{5m}$ as described in the previous case.
%The key idea is that, in the strategy of the defender for $H'$, all guards move in the same direction (left, right, up, down, left-down, etc).
%Moreover, all guards in a given column or line have the same position modulo 5.
%So we construct an orientation $H$ of $G$ as follows.
%We orientate each edge between $v_{i, j}$ and $v_{i,j+1}$ for $i \in [5n]$, $j \in [5n-1]$ 
%and between $v_{i,j}$ and $v_{i+1, j}$ for $i \in [5n-1]$, $j \in [5m]$ in same direction as in $H'$.
%If $x > 0$, we decompose the vertices $v_{i, j}$ with $i \in [5n+1, 5n+x]$ and $j \in [5m]$
%into blocks of size $(x, 5)$ and we orientate such that each block induces a strongly connected graph.
%We orientate such that $H$ contains the following edges for every $i \in [5m]$:
%$(v_{5n, i}, v_{5n+1, i})$, $(v_{5n+1,i+1}, v_{5n, i})$, $(v_{5n+x, i},v_{1, i})$, $(v_{(1, i}, v_{5n+x, i+1})$.
%We do the same procedure symmetrically if $y > 0$ and we orientate the remaining edge
%arbitrarily. TODO
\end{proof}

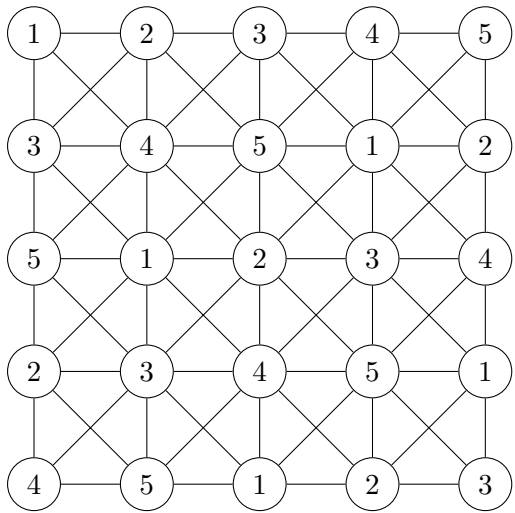
\begin{figure}
\begin{center}
\begin{tikzpicture}[vertex/.style={circle,draw,minimum size=20}]
  \foreach \x in {0,...,4}
    \foreach \y in {0,...,4} {
       \pgfmathtruncatemacro{\label}{mod(\x + 3 * \y + 3,5)+1}
       \node [vertex]  (\x\y) at (1.5*\x,1.5*\y) {\label}; 
    }

  \foreach \x in {0,...,4}
    \foreach \y [count=\yi] in {0,...,3}
      \draw (\x\y)--(\x\yi) (\y\x)--(\yi\x) ;

  \foreach \x  [count=\xi] in {0,...,3}
    \foreach \y [count=\yi] in {0,...,3}  
      \draw (\x\y)--(\xi\yi) (\yi\x)--(\y\xi) ;

\end{tikzpicture}
\caption{$(5, 2)$-PE coloring of a square of a king's grid}
\label{necoloring-king}
\end{center}
\end{figure}

Notice that we can obtain an upper bound $\frac{nm}{5} + O(n+m)$ when there is no condition on $n$ and $m$. The idea is similar to the proof of Theorem \ref{oednm-toric-grids}. However, the proof is quite complicated and the result not essential so we omit it in this paper.

We also generalize Theorem \ref{oednm-toric-grids} to toroidal hypergrids.

\begin{theorem}\label{oednm-toric-hypergrids} 
$\oednm(C_{n_1} \square \ldots \square C_{n_k}) \leq \frac{n}{k+1}$ where $n$ is the order of the graph and all $n_i$ are multiples of $k+1$.
\end{theorem}

\begin{proof}
Let $v$ be a vertex at position $(i_1, \ldots, i_k)$ in the hypergrid.
We affect to $v$ the color $$\sum_{j=1}^{k} j i_j \mod (k+1).$$
It is easily seen that this coloring is proper.
Additionally, for every distinct colors $i, j$ and vertex $v$ of color $i$,
$v$ has exactly two neighbors of color $j$.
Indeed, if $v$ is at position $(i_1, \ldots, i_k)$,
then the two neighbors are at positions
$(i_1, \ldots, i_{p-1}, i_p + 1, i_{p+1}, \ldots, i_k)$
where $p = j - i \mod (k + 1)$
and $(i_1, \ldots, i_{q-1}, i_q - 1, i_{q+1}, \ldots, i_k)$
where $q = i - j \mod (k + 1)$.
Thus, we obtain a $(k+1, 2)$-NE coloring.
\end{proof}

We conjecture that the upper bounds in Theorems \ref{oednm-toric-kinggrids}, \ref{oednm-toric-hypergrids} and \ref{oednm-toric-grids} correspond to the exact value of $\oednm$.
A way to prove this would be to show that that any orientation that minimizes $\ednm$ is eulerian.
More generally, we think that the following proposition is true.

\begin{conjecture}
Let $G$ be a graph that admits a $(k, 2)$-NE coloring.
Then $\oednm(G) = \frac{n}{k}$.
\end{conjecture}

This conjecture is verified for even cycles (Theorem \ref{edn_cycle})  and for square rook's grids (Theorem \ref{oednm_rook}).

\section{Future works and open questions}

Besides the two conjectures given in this paper, we enumerate some future works and open questions.

\begin{itemize}
\item
Give a tight upper bound of $\oednm$ depending on $n$ for 2-edge-connected graphs.
\item
Is there a natural parameter for digraphs that is an upper bound of $\edn$ as the clique covering number is for graphs? 
\item
Give better bounds for $\oedn$ on complete graphs and grids.
\item
We have proved that $\oedn = \oalpha$ for trees, cycles, complete bipartite graphs and grids $2\times n$, $3\times n$, $4\times n$. Is it true for complete graphs or (rectangular) grids in general?
\item
Can we characterize the graphs for which $\oednm = \gamma$?
Such examples of graphs are rook's graphs (Theorem \ref{oednm_rook}) and non complete graphs with $\oednm = 2$ (Theorem \ref{oednequals2}).
Remember that the only graphs for which $\oedn = \gamma$ are the graphs without edges (Proposition \ref{eternal_alpha_olpha}).
\item
Extend the study on trivially perfect graphs to cographs.
\item
Study the complexity of deciding whether $\oednm(G) \leq k$ in the general case and when $k$ is fixed.
Notice that, for $\oedn$, the problem is coNP-hard in the general case  (Corollary \ref{oedn-coNP}) 
and trivial  when $k$ is fixed.
Indeed, thanks to Corollary \ref{edn_bounds_on_clique}, there is only a finite number of positive instances.
\end{itemize}

\bibliographystyle{abbrv}
\bibliography{biblio}
\end{document}